\documentclass[11pt,letterpaper]{article}
\usepackage{fullpage}
\usepackage{amsmath,amssymb,amsthm,mathrsfs,mathabx}
\usepackage{dsfont}
\usepackage[normalem]{ulem}
\usepackage{hyperref}

\usepackage[svgnames]{xcolor}
\definecolor{links}{RGB}{11, 85, 255}
\definecolor{cites}{RGB}{0, 200, 0}
\definecolor{urls}{RGB}{255, 116, 0}
\hypersetup{colorlinks={true},linkcolor={links},citecolor=[named]{cites},urlcolor={urls}}

\usepackage[numbers,sort&compress]{natbib}
\usepackage[ruled,noend,noline]{algorithm2e} 
\usepackage[capitalise,noabbrev,nameinlink]{cleveref}

\crefname{ineq}{Inequality}{inequalities} 
\creflabelformat{ineq}{#2{\upshape(#1)}#3} 

\crefname{form}{Formula}{formulae} 
\creflabelformat{form}{#2{\upshape(#1)}#3} 

\newtheorem{example}{Example}
\newtheorem{theorem}{Theorem}
\newtheorem{lemma}{Lemma}
\newtheorem{proposition}{Proposition}

\DeclareMathOperator*{\ex}{\mathbb{E}}
\newcommand{\E}[2]{\ex_{#1}\left[#2\right]} 
\newcommand{\CondE}[3]{\ex_{#1}\left[#2 \;\middle\vert\; #3\right]} 

\newcommand{\fwhs}[1]{\; | \; #1 }

\newcommand{\prob}[1]{\ensuremath{\mathbb{P}\left(#1\right)}}

\DeclareMathOperator*{\argmax}{argmax}

\DeclareMathOperator*{\expectation}{\mathbb E}

\newcommand{\expect}[1]{\ensuremath{\expectation\left[#1\right]}}

\newcommand{\indicate}[1]{\mathbb{I}_{\left\{#1\right\}}}

\newcommand{\seller}{v_s}

\newcommand{\bP}{\mathbb{P}}
\newcommand{\trade}{\text{trade}}
\allowdisplaybreaks

\begin{document}

\title{Efficient Two-Sided Markets with Limited Information\thanks{Part of the work was done while Paul D\"utting was an Associate Professor in the Department of Mathematics at London School of Economics. The work of Federico Fusco, Philip Lazos, Stefano Leonardi, and Rebecca Reiffenh\"auser was supported in part by ERC Advanced Grant 788893 AMDROMA ``Algorithmic and Mechanism Design Research in Online Markets'' and MIUR PRIN project ALGADIMAR ``Algorithms, Games, and Digital Markets.''
}}

\newcommand*\samethanks[1][\value{footnote}]{\footnotemark[#1]}
\author{
Paul D\"utting\thanks{Google Research, Brandschenkestrasse 110, 8002 Z\"urich, Switzerland.  Email: \texttt{\href{mailto:duetting@google.com}{duetting@google.com}}}
\and
Federico Fusco\thanks{Department of Computer, Control, and Management Engineering ``Antonio Ruberti,'' Sapienza University of Rome, Via Ariosto 25, 00185 Rome, Italy.
Email: \texttt{\href{mailto:federico.fusco@uniroma1.it}{federico.fusco@uniroma1.it},
\{\href{mailto:lazos@diag.uniroma1.it}{lazos},
\href{mailto:leonardi@diag.uniroma1.it}{leonardi},
\href{mailto:rebeccar@diag.uniroma1.it}{rebeccar}\}@diag.uni\allowbreak roma1.it}
}
\and
Philip Lazos\samethanks
\and
Stefano Leonardi\samethanks
\and
Rebecca Reiffenh\"auser\samethanks
}

\date{}

\maketitle

\begin{abstract}
    A celebrated impossibility result by Myerson and Satterthwaite \cite{myersonS83}  
    shows that any truthful mechanism for two-sided markets that maximizes social welfare must run a deficit, resulting in a necessity to relax welfare efficiency and the use of approximation mechanisms. Such mechanisms in general make extensive use of the Bayesian priors.  In this work, we investigate a question of increasing theoretical and practical importance: how much prior information is required to design mechanisms with near-optimal approximations?
    
    Our first contribution is a more general impossibility result stating that no meaningful approximation is possible without any prior information, expanding the famous impossibility result of Myerson and Satterthwaite.
    
    Our second contribution is that one {\em single sample} (one number per item), arguably a minimum-possible amount of prior information,  from each seller distribution is sufficient for a large class of two-sided markets. We prove matching upper and lower bounds on the best approximation that can be obtained with one single sample for subadditive buyers and additive sellers, regardless of computational considerations. 
    
    Our third contribution is the design of computationally efficient blackbox reductions that turn any one-sided mechanism into a two-sided mechanism with a small loss in the approximation, while using only one single sample from each seller.
    On the way, our blackbox-type mechanisms deliver several interesting positive results in their own right, often beating even the state of the art that uses full prior information.
\end{abstract}

\section{Introduction}

Two-sided markets play an increasingly important role in practice, with applications ranging from internet trading platforms (such as eBay), to ridesharing platforms (such as Uber or Lyft), to ad exchanges (such as AppNexus). A canonical two-sided market consists of sellers that provide certain services or goods to the market, and buyers that consume these services or goods. It is natural to assume that both sides of the market have private information about their production costs and willingness to pay.

The design of such two-sided markets with private 
information has a long tradition in economics \cite{vickrey61, myersonS83}.
One typically strives to find mechanisms that guarantee individual rationality (IR), dominant strategy incentive compatibility (DSIC), budget balance (BB) (i.e., that the mechanism collects more money from the buyers than it gives to the sellers), all while also maximizing social welfare.
Designing such mechanisms, however, turns out to be rather challenging. Consider, for example, the celebrated Vickrey-Clarke-Groves (VCG) mechanism \cite{vickrey61, clarke71, groves73}. 
This mechanism chooses an allocation that maximizes social welfare, and charges each agent the amount by which the agent's presence reduces the social welfare of the other agents. Already Vickrey \cite{vickrey61} observed that when applied to two-sided markets, this mechanism is IR, DSIC, and maximizes social welfare, but may produce a deficit.

\begin{example}[Two-sided VCG for a single good] \label{ex:vcg} 
Consider a seller owning a single indivisible good, which she values at $v_s$, and two interested buyers willing to pay $v_{b_1}$ and $v_{b_2}$, respectively. Assume $v_s<v_{b_1}>v_{b_2}$, so that the socially efficient outcome is to give the item to $b_1$, achieving a welfare of $v_{b_1}$. The VCG mechanism will assign the item to $b_1$, whose participation keeps the item from being assigned to $s$ or $b_2$, so she will pay $\max\{v_s, v_{b_2}\}$. For the seller $s$, in her absence no one can get the item and welfare is reduced by $v_{b_1}$, so the mechanism will pay her this amount. Recall that $\max\{v_s,v_{b_2}\}-v_{b_1}<0$, so the mechanism runs a deficit.
\end{example}

This example is actually fairly robust, extending to settings with Bayesian priors on multiple buyers and multiple sellers. A famous impossibility result by \citet{myersonS83} in fact shows that this is not just a problem of VCG, it is a more fundamental incompatibility: any mechanism that is IR and DSIC and maximizes social welfare must violate budget balance even if the mechanism knows the Bayesian priors of the buyers and sellers. 
In this paper, we explore the fundamental information-theoretic as well as computational requirements for achieving near-optimal IR, DSIC, and BB mechanisms.

\subsection{
Information-Theoretic Barrier: Necessity of Some Prior Dependence}
Since information about the underlying distributions of the agents' valuations is often scarce, it would be ideal to find an approximation mechanism that guarantees close to optimal welfare without the use of such prior information. However, this is impossible: extending the famous inefficiency theorem of Myerson and Satterthwaite \cite{myersonS83}, we prove that without prior information, the welfare cannot be approximated up to any constant:

\begin{theorem}[Impossibility; \cref{sec:impossibility}]
\label{thm:impossibility}
No mechanism $M$ for bilateral trade (or richer variants of it) without prior information, where agents' valuations for a single item take arbitrary values in $[0,\,\infty)$, can be IC (in expectation), IR, BB and also guarantee an $\alpha$-approximation to the optimal social welfare, for any fixed value $\alpha>1$. 
\end{theorem}
As mentioned above, this crucial necessity of prior information is bad news: even when there is some information available on the priors, it usually presents itself as a limited amount of incomplete
data and assuming to indeed know the underlying distributions exactly is often unrealistic.

\subsection{Information-Theoretically Optimal Mechanisms}\label{sec:adjusted-vcg}

In light of our new impossibility result (\cref{thm:impossibility}), an important question is how much information about the priors is needed to enable approximately welfare-maximizing IR, DSIC, and BB mechanisms. Various forms of limited access to the priors are conceivable. A standard vehicle in Algorithmic Game Theory---and the lens we adopt here---is sample complexity (e.g., \cite{ColeR14, DhangwatnotaiRY15,MorgensternR15,AzarKW14,CorreaDFS19, RubinsteinWW20}). That is, we assume that the mechanism designer has access to samples, possibly from both the buyers' and sellers' distributions, and can use this knowledge when setting up the mechanism.

Our first result is that for the very general class of \emph{subadditive} buyers and additive sellers, a \emph{single sample} from each seller (and \emph{no samples} from the buyers) suffices for a factor 2 appproximation to the optimal social welfare.

To this end we consider the following \emph{adjusted-objective VCG mechanism}: it takes a single sample from each seller distribution, and offers the sellers the sample as price.
Afterwards, it runs the (one-sided) VCG on the buyers and items of sellers who accepted their price.
But instead of considering the usual objective of welfare maximization, it selects an allocation and payments based on the adjusted welfare (= total value of the buyers minus sampled value of the respective items).\footnote{
This idea of a penalty modifying the objective of the VCG mechanism is also used by Blumrosen and Dobzinski \cite{blumrosen2014reallocation} (however with the median instead of a sample), which directly applied to our problems results in an $8-$approximation given subadditive buyer valuations, and with some modification to the proof one can show a factor of $4$.} 

\begin{theorem}[Adjusted VCG; \cref{sec:2subadditive}]
\label{thm:2-subadditive}
    For subadditive buyer valuations and unit-supply/additive sellers, the social welfare output by the adjusted-objective VCG mechanism which uses a single sample from each seller and no samples from the buyers yields a 2-approximation to the optimal welfare, in expectation. The mechanism is IR, DSIC, and BB.
\end{theorem}

This result is in fact fairly robust, and generally needs only a certain form of ``approximate location oracle'' for the seller distributions: a single sample is a particular economic way to obtain such an oracle, but other aggregate statistics such as the median or other quantiles (possibly at the expense of worse approximation guarantees) also suffice.

\cref{thm:2-subadditive} yields poly-time mechanisms for settings where the adjusted VCG can be computed efficiently. This includes gross substitutes valuations for the buyers \cite{NisanS06}, so we obtain an IR, DSIC, BB $2$-approximation for this class of valuations with a single sample from each seller. Another special case is bilateral trade, where our theorem implies the existence of a $2$-approximate IR, DSIC, and BB single-sample mechanism.

We complement our positive result concerning the adjusted VCG mechanism with the following lower bound, which shows that the factor $2$ is best possible for mechanisms that---just as the adjusted VCG mechanism---are \emph{deterministic}, in that the only randomness in these mechanisms stems from the randomness in the sample they receive. 

\begin{theorem}[Lower Bound; \cref{sec:2subadditive}]\label{thm:lower-bound}
    Every deterministic IC, IR and BB mechanism for bilateral trade that receives as sole prior information a single sample from the seller's distribution has approximation ratio at least $2$. 
\end{theorem}

We note that~\citet{blumrosen2014reallocation,DobzinskiB16} previously showed a factor $2$-approximate IR, DSIC, and BB mechanism for bilateral trade---the \emph{median mechanism}---and that no deterministic mechanism that uses only information from one side of the market can yield a better approximation.\footnote{Our lower bound is similar, but it does not immediately follow from their result. The reason is that mechanisms that use only one sample cannot be simulated by a deterministic mechanism using prior information.}
For more general valuations, the state of the art is a $6$-approximation for XOS valuations \cite{Colini-BaldeschiEtAl16,Colini-BaldeschiEtAl17} and a $8$-approximation for subadditive valuations \cite{blumrosen2014reallocation}.

We extend the factor $2$ of \citet{DobzinskiB16} for bilateral trade to single-sample mechanisms, and the much more general problem with subadditive buyers and additive sellers.

While these results nicely illuminate the information-theoretical situation, the approach of adjusting the objective of a one-sided mechanism to accomodate seller payments has important limitations. Essentially, it can be applied only via (generalized) VCG because of its properties as an \emph{affine maximizer}, but VCG is not necessarily computationally feasible and cannot be used, e.g., in online environments.

\subsection{Poly-Time Near-Optimal Single-Sample Mechanisms}

We next address the shortcomings and limited flexibility of the adjusted objective approach. In general, it is not even clear what it means ``to run a given one-sided approximation mechanism on the adjusted objective.''
We will instead use that for several subclasses of subadditive functions (such as XOS or GS) we can adjust buyer valuations by deducting item prices, and the resulting valuations will still be within that same class.
We leverage this as follows: As in our approach in \cref{sec:adjusted-vcg} we use a single sample from each seller distribution. We offer this sampled price to the sellers. Then we run the given one-sided approximation mechanism on the buyers and the items of those sellers who accepted their sampled price, with buyer valuations reduced by the respective item prices. We charge the buyers whatever the one-sided mechanism would charge them plus the sum of the sampled prices of the items they receive; the sellers whose items were assigned to buyers receive the sampled price.

The constructed mechanism is BB by design. We show that if the one-sided mechanism is DSIC then so is the two-sided mechanism. From the approximation guarantee of the one-sided mechanism we get a guarantee with respect to the buyer surplus (w.r.t. the original valuations and sampled prices), and we show that this translates (approximately) to the actual objective we are interested in. Moreover, if the one-sided mechanism is online, then so is the two-sided mechanism (for the buyers).

\begin{theorem}[Black Box I; \cref{section:combinatorial}]
\label{thm:black-box-1}
Denote by $\alpha$ the approximation guarantee of any one-sided IR, DSIC offline/online mechanism for maximizing social welfare for XOS valuations. We give a two-sided mechanism for XOS buyers and unit-supply sellers that is IR, DSIC, BB, uses a single sample from each seller and provides a 
$\max\{2\alpha,3\}$
approximation to the optimal social welfare. The two-sided mechanism inherits the offline/online properties of the one-sided mechanism on the buyer side and is offline on the seller side.
\end{theorem}

Two implications of \cref{thm:black-box-1} are: for matching constraints (unit demand valuations), we can apply the poly-time truthful secretary algorithm by \citet{Reiffenhauser19}; this yields a IR, DSIC, and BB $2e$-approximate single-sample mechanism that is online random order on the buyer side and offline on the seller side. 
For general XOS functions we obtain a poly-time IR, DSIC, and BB $O((\log \log m)^3)$-approximate single-sample mechanism by plugging in the state of the art truthful approximation mechanism for this class \cite{Dobzinski16,AssadiS19,AssadiEtAl21}.
Both of these are the first single-sample mechanisms for the respective settings. They compare to a $6$-approximation which can process buyers in any order, but requires additional information \cite{Colini-BaldeschiEtAl16,Colini-BaldeschiEtAl17}.
 
Our second reduction is more restricted since it applies only to unit-demand buyers and unit-supply sellers with identical items, but it improves on the first reduction in two ways: it yields a strongly BB mechanism in which the total payments collected from the buyers exactly match those to the sellers, and it additionally allows online random arrival of the sellers. 

The reduction runs the given one-sided mechanism to select buyers and tentative payments for the buyers, it draws a single sample from each seller distribution, and then randomly matches the selected buyers and sellers offering them to trade at the maximum of the two prices.

\begin{theorem}[Black Box II; \cref{section:double}]
\label{thm:black-box-2}
Denote by $\alpha$ the approximation guarantee of any one-sided IR, DSIC offline/online mechanisms for the intersection of a downward closed set system $\mathcal{I}$ with a uniform matroid that may or may not use information on the priors.
We give a two-sided mechanism for unit-demand buyers, unit-supply sellers with identical goods, and constraints $\mathcal{I}$ on the buyers, that is IR, DSIC,  SBB, has the same information requirements on the buyer side as the one-sided mechanism and uses a single sample from each seller, and yields a 
$(1 + 1/(2-\sqrt{3})\cdot \alpha) \approx (1+3.73 \cdot \alpha)$
approximation to the optimal social welfare. The mechanism inherits the online/offline properties of the one-sided mechanism on the buyer side and it is online random order on the seller side.
\end{theorem}

\cref{thm:black-box-2} leads to  additional results for double auctions.   We specifically obtain a $1+3.73(1+o(1))$-approximate mechanism for a $k$-uniform matroid constraint on  the buyers by using the truthful secretary algorithm by \citet{Kleinberg05}. The mechanism admits random order arrivals of buyers and sellers, and requires a single sample on the seller side.  We also obtain a $1+3.73(1+o(1))$-approximate mechanism for  a $k$-uniform matroid with fixed order arrivals of the buyers and random order arrivals from the sellers by using the single sample prophet inequality algorithm by \citet{AzarKW14}.  This second mechanism requires a single sample from each buyer and seller.  (All our other mechanisms only use samples from the sellers.)

Note that besides this variety of results for different valuation classes, arrival models, and even prior information assumptions on the buyer side, any future contribution regarding the one-sided problems will also yield a new result for the two-sided version via our constructions.

\subsection{Related Work}

A landmark result in economics/mechanism design is Myerson and Satterthwaite's impossibility result \cite{myersonS83}, which states that there can be no IC (BIC or DSIC) mechanism that is budget balanced and maximizes social welfare. The same paper also describes the ``second best'' mechanism which maximizes social welfare subject to the other constraints. 

In a sequence of papers Satterthwaite et al.~\cite{SW89,rsw94,sw02} considered the non-truthful \emph{buyer's bid mechanism} for double auction settings with identical items, and showed that with i.i.d.~buyers and i.i.d.~sellers as the number of buyers and sellers $n,m \rightarrow \infty$ equilibrium bids converge to truthful bids and the corresponding allocation converges to an efficient allocation. 
In a very elegant paper, McAfee \cite{mcafee92} proposed the DSIC and budget balanced \emph{trade reduction mechanism} for double auction settings with identical items. This mechanism allows all but the least efficient trade. Thus, when $\ell$ is the number of trades in an efficient solution, it achieves a $1-1/\ell$ approximation.
McAfee's trade reduction idea has since been extended to a variety of considerably more general settings (see, e.g.,  \cite{BabaioffN04,BabaioffW05,DRT17}). 

A pioneering paper by Blumrosen and Dobzinski \cite{blumrosen2014reallocation} proposed the \emph{median mechanism} for bilateral trade, which sets a posted price equal to the median of the seller's distribution and shows that this mechanism obtains a $2$-approximation to the optimal welfare. Subsequent work by the same authors \cite{DobzinskiB16} has shown that this factor $2$ is actually optimal for deterministic mechanisms that only use information about the seller distribution, and improved the approximation guarantee to $e/(e-1)$ through a randomized mechanism whose prices depend on the seller distribution in a more intricate way. Colini-Baldeschi et al.~\cite{Colini-BaldeschiEtAl16} gave a fixed-price mechanism that achieves a $25/13 \approx 1.92$ approximation, and show a lower bound of $\approx 1.3360$ that applies to any DSIC and strongly budget balanced mechanism. Most recently, Kang and Vondr\'ak  \cite{kang2018strategy} gave a $e/(e-1)-\epsilon$ approximate mechanism for bilateral trade.

Colini-Baldeschi et al.~\cite{Colini-BaldeschiEtAl16,Colini-BaldeschiEtAl17} also develop DSIC and strongly budget balanced mechanisms for more general settings with combinatorial buyers and unit-supply sellers, using ideas from prophet inequalities and posted-price mechanisms for one-sided markets. 
The first paper gives a $16$-approximate DSIC and strongly budget balanced mechanism for double auction settings where the buyers can be subject to a matroid constraint. This mechanism is online fixed order on the buyer side, and online random order on the seller side.
The second paper gives a $6$-approximate DSIC and strongly budget balanced mechanism for XOS buyers and unit-supply sellers that is online fixed order on the buyer side and offline on the seller side. 

A parallel line of work \cite{HHA16,Blumrosen2016,mcafee2008gains,DBLP:conf/sigecom/BrustleCWZ17,colini2017fixed,Babaioff0GZ18} has considered the related but different objective of optimizing the \emph{gain from trade}, which measures the expected increase in total value that is achievable by applying the mechanism, with respect to the initial allocation to the sellers.  Gain from trade is harder to approximate than social welfare, and $O(1)$ approximations of the optimal Bayesian mechanism are only possible in BIC implementations.

\citet{GoldnerEtAl20} recently suggested an alternative, resource augmentation approach to gains from trades in two-sided markets.
They ask how many buyers (resp.~sellers) need to be added into the market so that a variant of McAfee's trade reduction mechanism yields a gain from trade superior to the optimal gains from trade in the original market. As a side product they obtain a $4$-approximate single-sample mechanism for gains from trade under natural conditions on the distributions.

Another related line of work considers the problem of maximizing revenue in double auction settings, either in static environments \cite{GomesM14,NiazadehYK14} or in dynamic environments \cite{BalseiroMLZ19}.
A variation where both buyers and sellers arrive dynamically and the mechanism can hold on to items was investigated in \cite{giannakopoulos2017online,koutsoupias2018online}.
On a higher level, our work is related to prior work that has shown how having a single sample enables the design of mechanisms that achieve near-optimal revenue (e.g., \cite{DhangwatnotaiRY15}) or prophet inequalities and posted-price mechanisms with near optimal welfare (e.g., \cite{AzarKW14,CorreaDFS19,CorreaCES20}). 
It is also related to prior work on composing mechanisms in an incentive compatible manner, such as \cite{MN08} for one-sided markets and \cite{DRT17} for two-sided markets.

\section{Model and Definitions}

\paragraph{Two-Sided Markets.}
In a two-sided market we are given a set of $n$ buyers $\mathcal{B}$ and a set of $m$ sellers $\mathcal{S}$. Each seller $s$ is {\em unit supply}, i.e., she has a single indivisible item for sale, and a private valuation $v_{s} \in \mathbb{R}_{\geq 0}$ for the item she sells. Each buyer $b$ has a private valuation function $v_{b}: 2^{\mathcal{S}} \rightarrow \mathbb{R}_{\geq 0}$, mapping each set of sellers to a non-negative real. All results hold also for the equivalent case of sellers with additive valuations.
We write $v_\mathcal{B}$ and $v_\mathcal{S}$ for the vector of buyer valuations and seller valuations, respectively. 
Buyer and seller valuations are drawn independently from distributions $F_{b}$ for $b \in \mathcal{B}$ and $F_{s}$ for $s \in \mathcal{S}$.

The valuation functions of the buyers will be constrained to come from some class of functions $\mathcal{V}$. Buyers are \emph{unit demand} if for each buyer $b$ and set of sellers $S$, $v_{b}(S) = \max_{s \in S} \{v_{b}(\{s\})\}$. 
Buyers have \emph{subadditive} valuations if for each $b$ and sets of sellers $S_1,S_2$, $v_b(S_1\cup S_2)\leq v_b(S_1)+v_b(S_2)$.
Buyers have \emph{fractionally subadditive} (or XOS) valuations if for each buyer $b$ and every set of sellers $S$, $v_{b}(S) = \max_{a \in A_{b}} \sum_{s \in S} a(s)$, where $A_{b}$ is a set of additive valuation functions. We say that items are \emph{identical} if the valuation function $v_{b}(S)$ of all buyers $b$ only depends on the cardinality of the set $S$ they receive, i.e., for all $b$ and all $S, S' \subseteq \mathcal{S}$ with $|S| = |S'|$ we have $v_{b}(S) = v_{b}(S')$. 
Otherwise, items are \emph{non-identical}.

We also allow for constraints on which buyers can trade simultaneously. We express these constraints through set systems $\mathcal{I}_\mathcal{B} \subseteq 2^\mathcal{B}$. We require these set systems $\mathcal{I}_{\mathcal{B}}$ to be downward closed. That is, whenever $X \subseteq Y$ and $Y \in \mathcal{I}_{\mathcal{B}}$, then also $X \in \mathcal{I}_{\mathcal{B}}$. 
Of particular interest for our work will be \emph{matroids}, i.e., downward-closed set systems that additionally satisfy a natural exchange property. Formally: whenever $A,B \in \mathcal{I}$ and $|A| > |B|$ then there exists $a \in A \setminus B$ such that $B \cup \{a\} \in \mathcal{I}$. A special case are \emph{$k$-uniform matroids} where $A \in \mathcal{I}$ whenever $|A| \leq k$.

An \emph{allocation} is a partition of the sellers $\mathcal{S}$ into $n$ disjoint sets $(S_1,\dots,S_n)$, i.e., $\bigcup_i S_i \subseteq \mathcal{S}$ and $S_i \cap S_j = \emptyset$ for all $i \neq j$, with the interpretation that buyer $b_i$ for $1 \leq i \leq n$ receives the items of the sellers in $S_i$. 
An allocation $A = (S_1, \dots, S_n)$ is \emph{feasible} if the set of buyers $\mathcal{B}_A = \{b_i \in \mathcal{B} \mid S_i \neq \emptyset\}$ that receive a non-empty allocation is admissible (i.e., $\mathcal{B}_S \in \mathcal{I}_\mathcal{B}$). 
The \emph{social welfare} of an allocation $A = (S_1,\dots,S_n)$ is given by the sum of the valuations that buyers $b_i$ for $1 \leq i \leq n$ have for the items of the sellers in their respective sets $S_i$ \emph{plus} the valuations of the sellers that are not assigned to any buyer, i.e.,
\[
\mathsf{SW}(A) = \mathsf{SW}(S_1, \dots, S_n) = \sum_{b_i \in \mathcal{B}} v_{b_i}(S_i) + \sum_{s \in \mathcal{S}, s \not\in \bigcup_i S_i} v_{s}.
\]
We use $\textsf{OPT}(v_\mathcal{B},v_\mathcal{S})$ to denote the feasible allocation that maximizes social welfare. 

\paragraph{Mechanisms.}
A (direct revelation) \emph{mechanism} $M = (x,p)$ receives bids $bid_{b}: 2^{\mathcal{S}} \rightarrow \mathbb{R}_{\geq 0}$ from each buyer $b \in \mathcal{B}$ and $bid_{s} \in \mathbb{R}_{\geq 0}$ from each seller $s \in \mathcal{S}$. The bids of the buyers
are constrained to be consistent with the class of functions $\mathcal{V}$ of their valuations. Bids represent reported valuations, and need not be truthful. In analogy to our notation for valuations, we use $bid_\mathcal{B}$ and $bid_\mathcal{S}$ for the vector of bids of all buyers or all sellers, respectively. 

A mechanism $M$ is defined through an \emph{allocation rule} $x: \mathcal{V}^{n} \times \mathbb{R}^{m}_{\geq 0} \rightarrow \bigtimes_{i=1}^n 2^{\mathcal{S}} $ and a \emph{payment rule} $p: \mathcal{V}^{n} \times \mathbb{R}^{m}_{\geq 0} \rightarrow \mathbb{R}^{n+m}$.   

The mapping from bids to feasible allocations can be randomized, in which case $x(bid_\mathcal{B},bid_\mathcal{S})$ is a random variable. The payments can also be randomized. 
We interpret the vector of payments as the payments that the buyers need to make to the mechanism, and that the sellers receive from the mechanism. 
We use the shorthand $p_{b}(bid_\mathcal{B},bid_\mathcal{S})$ and $p_{s}(bid_\mathcal{B},bid_\mathcal{S})$ to refer to the payment from buyer $b$ to the mechanism and from the mechanism to seller. We require that $p_{s} = 0$ if seller $s$ keeps her item. 

A mechanism is \emph{single-sample} if the only information it is given about the two sets of distributions $F_{b}$ for $b \in \mathcal{B}$ and $F_{s}$ for $s \in \mathcal{S}$ is a single sample from each of these distributions.

\paragraph{Utilities.} 
We assume that buyers and sellers have \emph{quasi-linear utilities}, and that they are utility maximizers. It means that the utility $u^M_{b}(v_{b}, (bid_\mathcal{B},bid_\mathcal{S}))$ of buyer $b$ with valuation function $v_{b}$ in mechanism $M = (x,p)$ under bids $(bid_{\mathcal{B}},bid_{\mathcal{S}})$ is given by her valuation for the items she receives minus payment.
That is,
\[
    u^M_{b}(v_{b}, (bid_\mathcal{B},bid_\mathcal{S})) = \E{}{ v_{b}(x_i(bid_\mathcal{B},bid_\mathcal{S})) - p_{b}(bid_\mathcal{B},bid_\mathcal{S})},
\] 
where the expectation is over the randomness in the mechanism.
The utility $u^M_{s}(v_{s}, (bid_\mathcal{B},bid_\mathcal{S}))$ of a seller $s$ in mechanism $M = (x,p)$ under bids $(bid_{\mathcal{B}},bid_{\mathcal{S}})$
is the payment she receives if she sells the item and her valuation for her item otherwise.
Formally, $u^M_{s}(v_{s},(bid_{\mathcal{B}},bid_{\mathcal{S}}))$ is equal to 
\[
\E{}{\indicate{\text{$s$'s item is sold}} \cdot p_{s} (bid_{\mathcal{B}},bid_{\mathcal{S}}) + \indicate{\text{$s$'s item is not sold}} \cdot v_{s}},  
\]
where the expectation is over the randomness in the mechanism.

\paragraph{Goals.}
We seek to design mechanisms, and specifically single-sample mechanisms, with the following desirable properties:

\medskip\noindent(1) {\bf Individual Rationality.} 
Mechanism $M = (x,p)$ is \emph{individually rational (IR)} if for all $b \in \mathcal{B}$
$u^M_{b}(v_{b},(v_{\mathcal{B}},v_{\mathcal{S}})) \geq 0$ and for all $s \in \mathcal{S}$, $u^M_{s}(v_{s},(v_{\mathcal{B}},v_{\mathcal{S}})) \geq v_{s}.$ 

\medskip\noindent(2) {\bf Incentive Compatibility.} 
Mechanism $M = (x,p)$ is \emph{(dominant-strategy) incentive compatible (DSIC)} or \emph{truthful} if
for each buyer $b$ and each seller $s$, all valuation functions $v_{b}$ and $v_{s}$, all possible bids
$bid_{\mathcal{B}}$ by the buyers, and all possible bids by the sellers $bid_{\mathcal{S}}$, it holds that 
\begin{align*}
u_{b}^M(v_{b},((v_{b},bid_{\mathcal{B}\setminus\{b\} }),bid_\mathcal{S})) &\geq u_{b}^M(v_{b},(bid_{\mathcal{B}},bid_{\mathcal{S}})), 
\\
u^M_{s}(v_{s},(bid_{\mathcal{B}},(v_{s},bid_{\mathcal{S}\setminus \{s\}}))) &\geq u^M_{s}(v_{s},(bid_{\mathcal{B}},bid_{\mathcal{S}})),
\end{align*}
where $bid_{\mathcal{B}\setminus \{b\}}$ denotes the set of all the buyer bids but $b$'s and $bid_{\mathcal{S}\setminus \{s\}}$ denotes the set of all the seller bids but $s$'s.

\medskip\noindent(3) {\bf Budget Balance.} 
A truthful mechanism $M = (x,p)$ is \emph{weakly budget balanced (BB or WBB)} if 
\[
\E{}{\sum_{b \in \mathcal{B}} p_{b}(v_{\mathcal{B}},v_{\mathcal{S}}) - \sum_{s\in\mathcal{S}} p_{s}(v_{\mathcal{B}},v_{\mathcal{S}})} \geq 0,
\] and it is \emph{strongly budget balanced (SBB)} if the above holds with equality.

\medskip\noindent(4) {\bf Efficiency.} 
Finally, a truthful mechanism $M = (x,p)$ provides an \emph{$\alpha$-approximation} to the optimal social welfare, for some $\alpha \geq 1$, if it holds that $ \alpha \cdot \mathbb{E}[\mathsf{SW}(x(v_\mathcal{B},v_{\mathcal{S}}))] \geq \mathbb{E}[\mathsf{SW}(\mathsf{OPT}(v_\mathcal{B},v_\mathcal{S}))].$

\paragraph{Remark.} 
Our mechanisms will actually satisfy even stronger IR and BB properties in that they will satisfy these conditions ``ex post'' (i.e., pointwise).

\section{Impossibility of Approximation without Prior Information}
\label{sec:impossibility}

In this section we show that at least a tiny amount of information about the priors is needed in order to obtain meaningful approximations of the social welfare with IR,  DSIC, and BB mechanisms. 

We present here the proof of \cref{thm:impossibility} for deterministic mechanisms, which nicely illustrates the difficulties at the core of this statement. The considerably more technical, detailed proof for the general theorem can be found in \cref{app:proofs_of_section3}.
\addtocounter{theorem}{-5}
\begin{theorem}[Impossibility]
No mechanism $M$ for bilateral trade (or richer variants of it) without prior information, where agents' valuations for a single item take arbitrary values in $[0,\,\infty)$, can be IC (in expectation), IR, WBB and also guarantee an $\alpha$-approximation to the optimal social welfare, for any fixed value $\alpha>1$. 
\end{theorem}

\begin{proof}
    Fix any constant approximation ratio $\alpha$, and assume that $M$ is an IR, weakly budget-balanced, truthful $\alpha$-approximation for the problem. 
    For the ease of exposition assume $M$ is deterministic: the same proof strategy works in the most general case.
    
    Assume now that $v_b>\alpha \cdot v_s$. The only way $M$ can guarantee this ratio is by trading the item from $s$ to $b$. Doing this truthfully and in an individually rational fashion without running a deficit implies that for the price $p_b$ charged from the buyer, and $p_s$ paid to the seller, it holds that $v_s \leq p_s \leq p_b \leq v_b$.
        
    Recall that intuitively, the only play here is to set $p_s=v_b$ and $p_b=v_s$. In fact, for truthfulness, an agent should not be able to influence the price she pays, and when determining said price, the only value we therefore know for sure will be in the feasible interval is the reported valuation of the \emph{opposite} agent. Our argument is close to this intuition, however a bit more intricate.

    We assumed $M$ to be truthful, i.e., if $v_s$ is a fixed value, then $b$'s utility will be at least as good when she reports the true $v_b$ as when reporting some $v'_b$.
    Since otherwise, whenever $v_b$ is actually the value with the higher price, the agent is incentivized to make a false report. Note that $M$ is forced to make a trade for both of the values in order to preserve an $\alpha$-approximation.
    We have established so far that the buyer price $p_b^{\star}$ will be a single value for all buyer reports in $(\alpha\cdot v_s,\,\infty)$, and for all of those buyer reports, there also needs to be a trade. By budget balance, when this is the case it holds $ p_b^{\star}\in [v_s,\,\alpha\cdot v_s]$ which means that even the smallest buyer for which $M$ has to trade will be willing to accept $p_b^{\star}$.

    This shows that whenever the seller reports a fixed value $v_s$, the price for each \emph{critical} buyer will come from a fixed interval starting at $v_s$.

     We can do the same from the other direction. Start fixing some buyer report $v_b$; since $M$ is an $\alpha$-approximation, it will have to make a trade whenever $v_s\in [0,v_b/\alpha)$. Because $M$ is truthful also for the sellers, each seller report in this interval must result in the same price $p_s^{\star}$. To be acceptable for all sellers in the critical interval, similarly as above, it must be that $p_s^{\star}\in [v_b/\alpha, v_b]$.
    We put this together with the fact that, due to weak budget balance, $p_s^{\star}$ cannot exceed $p_b^{\star}$, i.e., $       [v_b/\alpha,\, v_b]\ni p_s^{\star}\leq p_b^{\star} \in [v_s,\,\alpha \cdot v_s].$
    This goes awry for any choice of numbers $p_s^{\star},\, p_b^{\star}$ once $v_s$ and $v_b$ are such that the possible intervals for $p_s^{\star},\, p_b^{\star}$ do not overlap.
    Setting for example $v_s^{\star}= 1$ and $v_b^{\star}= \alpha^3$ violates the previous inequality.
\end{proof}

\section{Two-Approximation for Subadditive Buyers}
\label{sec:2subadditive}
In this section, we present a single-sample, IR, BB and DSIC $2$-approximation for subadditive buyers. Our mechanism is essentially the generalized VCG mechanism, which we use as a blackbox. 

Given a sample $v'_s$ for each seller, we do the following:
\begin{itemize}
    \item For each $s\in S$, offer a price $v'_s$ to $s$ and if she accepts, add $s$ to the set $\hat S$ of available items.
    \item On the instance $(B, \hat S)$, run VCG with the objective of maximizing the affine function 
    $\sum_{b\in B} v_b(S_b)- \sum_{s\in \bigcup_{b\in B}S_b}v'_s$,
    where $S_b$ denotes the items assigned to $b$.
    \item Assign the items in $\bigcup_{b\in B} S_b$ to their respective buyers, then charge them the price determined by VCG and pay each according seller a price of $v'_s$. All other sellers keep their item.
\end{itemize}
Note that VCG outputs an optimal assignment according to above objective because of its well-known properties as an affine maximizer. 
Showing that our mechanism is DSIC, IR, and BB is not too hard and has been done before, see \citet{blumrosen2014reallocation}
for more details. Intuitively it is IR because we ask the sellers if they want to trade and for the buyers, VCG is of course IR. It is DSIC for the sellers since they can refuse, and for the buyers because VCG is DSIC. 

Finally, it is BB because for every buyer $b$, the sum of her VCG payments equals at least the sum of her items' samples, which is what we pay to the sellers. This is true because the loss introduced to the other agents by $b$'s presence in the VCG routine is at least the sum of samples of items $S_b$: by not assigning those at all in a solution without $b$, we gain exactly this amount according to our adjusted objective.

Proving that the mechanism is also a 2-approximation is considerably more involved and will need some preparation.
For simplicity, we use notation for an allocation and its welfare interchangeably.
Let $OPT_{max}$ denote an allocation on $(B,S)$ which maximizes the adjusted objective
$$
\sum_{b\in B}\left(v_b\left(S^{OPT_{max}}_b\right)-\sum_{s\in S^{OPT_{max}}_b}\max\{v_s, v'_s\}\right).$$
Further, for any allocation $A$ over sets $(S,B)$ with valuations $v_s,\, v_b$, respectively, its corresponding \emph{welfare} is $\sum_{s\in S}v_s +\sum_{b\in B}T_b^A$, where $T_b^A= v_b(A(b))-\sum_{s\in A(b)} v_s$. For use in the proof, we further define the \emph{adjusted} buyer gains $\hat{T}_b^A=v_b(A(b))-\sum_{s\in A(b)}\max\{v_s,v'_s\}$.

In order to prove the approximation, we first show three preliminary statements which we then combine. The first one, \Cref{lem:folklore}, reports a folklore insight on subadditive functions and sampling. A proof is added for completeness. 
\begin{lemma}
\label{lem:folklore}
    Let $f:2^\mathcal{N} \to \mathbb{R}$ be a subadditive set function, and let $X$ be a random set of the base set $\mathcal{N}$ where each element appears independently with probability $\frac 12$, then $\E{}{f(X)} \ge \frac 12 f(\mathcal{N})$.
\end{lemma}
\begin{proof}
    Start noting that every subset $X'$ of $\mathcal{N}$ is realized with the same identical probability $p = 2^{-|\mathcal{N}|}$, i.e. $\prob{X = X'} = p$. Therefore we have the following:
    \begin{align*}
        \E{}{f(X)} &= \sum_{X' \subseteq \mathcal{N}} p \cdot f(X') = \frac 12 \sum_{X' \subseteq \mathcal{N}} p \cdot \left( f(X') + f(\mathcal{N} \setminus X') \right) \\
        &\ge \frac 12 \sum_{X' \subseteq \mathcal{N}} p \cdot f(\mathcal{N}) = \frac 12 f(\mathcal{N}).
    \end{align*}
    Notice that the inequality is where we use subadditivity. 
\end{proof}

The other two lemmas we need are more technical and specific to our problem. 
\cref{OptMax} states that the loss in buyer value when replacing seller valuations with $\max\{v_s,v'_s\}$ is not too large, i.e., the buyers' contribution to $OPT_{max}$ is closely related to their contribution to $OPT$.

\begin{lemma}\label{OptMax}
The following inequality holds true:
\[  \sum_{b\in B} \mathbb{E}\left[\hat{T}_b^{OPT_{max}}  \right] \geq  \sum_{b\in B} \mathbb{E}\left[T_b^{OPT} \right]-\sum_{s\in S} \mathbb{E}\left[\left(v'_s-v_s\right)_+ \right] \]
\end{lemma}
\begin{proof}
Define the set of items assigned to any buyer by the optimal assignment $OPT$ to be $S^{OPT}\subseteq S$. Let $B^{max}\subseteq B$ be those buyers for which their endowments $S^{OPT}_b\subseteq S^{OPT}$ satisfy
\[ v_b(S^{OPT}_b) > \sum_{s\in S^{OPT}_b}\max\{v_s,v'_s\} \]
and call the set of items sold by $OPT$ to these buyers $S^{max}$. Consider the assignment $A^+$ implied by sets $B^{max}$, $S^{max}$. This assignment clearly is feasible under the adjusted objective, which we can model for the purpose of finding an assignment by replacing all seller valuations with $\max\{v_s,v'_s\}$ in VCG. Recall we denote with $\hat T_b$ the according, adjusted buyer gains $T_b$. 
We have
\[
\sum_{b\in B\setminus B^{max}}\hat T^{A^+}_b = 0 \text{ , while }\sum_{b\in B\setminus B^{max}}T^{OPT}_b \leq \sum_{s\in S\setminus S^{max}}\left(v'_s-v_s\right)_+.
\]
For the other items, i.e., those in $S^{max}$, we have
\begin{align*}
     \sum_{b\in B^{max}} \hat T^{A^+}_b &= \sum_{b\in B^{max}}v_b(S^{OPT}_b)-\sum_{s\in S^{max}} \max\{v_s, v'_s\}\\
     &= \sum_{b\in B^{max}}T^{OPT}_b - \sum_{s\in S^{max}} \left(v'_s-v_s\right)_+.
\end{align*}
Putting both together, this yields
\[ \sum_{b\in B} \hat T^{A^+}_b \geq \sum_{b\in B} T^{OPT}_b - \sum_{s\in S}\left(v'_s-v_s\right)_+\]
and since $OPT_{max}$ maximizes the sum of $\hat T_b$, also $\sum_{b\in B}\hat T^{OPT_{max}}_b\geq \sum_{b\in B}\hat T^{A^+}_b$. 
\end{proof}

Our next lemma, \cref{AdjVCGMax}, will allow us to relate our adjusted VCG mechanism to $OPT_{max}$. Let us denote the output of adjusted-objective VCG as $ALG$. 

\begin{lemma}\label{AdjVCGMax}
The following inequality holds true:
\[\mathbb{E}\left[ \sum_{b\in B}T^{ALG}_b \right]\geq \frac 12 \mathbb{E}\left[\sum_{b\in B}\hat{T}^{OPT_{max}}_b \right] \]

\end{lemma}
\begin{proof}[Proof of \cref{AdjVCGMax}]

Fix buyer $b$ and her assigned set 
$S^{OPT_{max}}_b$. Let for each $b$ the assignment $S^{1/2}_b$ be constructed by erasing each $s\in S^{OPT_{max}}_b$ with probability $1/2$, independently from everything else. It holds by definition
\[ \hat{T}_b^{OPT_{max}}=v_b(S^{OPT_{max}}_b)-\sum_{s\in S^{OPT_{max}}_b} 
 \max\{v_s,v'_s\}.\]
With subadditivity and \Cref{lem:folklore}, this implies for the assignment defined by the $S^{1/2}_b$ and the according $\hat{T}^{1/2}$:
\[ \mathbb{E}\left[ \hat{T}_b^{1/2}\right] \geq \frac 12 v_b(S^{OPT_{max}}_b) - \frac 12 \sum_{s\in S^{OPT_{max}}_b} \max\{v_s,v'_s\} \]
since every item is simply erased with fixed probability $1/2$, in both the left and the right hand side. Note that the expectation is only with respect to the erasing process, while the valuations are fixed.

Let $ALG'$ be the solution $OPT_{max}$ restricted to those items for which $v_s\leq v'_s$, which is the case for each $s$ with probability at least $1/2$, independently.
Note that since $OPT_{max}$ only cares about $\max\{v_s,v_s'\}$, it is oblivious to which value realizes the maximum. We infer $\mathbb{E}[\hat{T}_b^{ALG'}]\geq\mathbb{E}[\hat{T}_b^{1/2}]$ for each $b$, because $ALG'$ can be seen as simply fixing a certain way of flipping the coins for $S_b^{1/2}$, and then retaining some extra items (alternatively, we could employ a random tie breaking).
Next, we note that every subset $S_b^{ALG'}$ assigned to a buyer in $ALG'$ is contained in $\hat S$, and  for the objective of our algorithm, assigning it is advantageous, i.e., yields nonnegative contribution.
$ALG$ might deviate from these assignments, but it holds that
\[ \sum_{b\in B}v_b(S_b^{ALG})-\sum_{s\in\bigcup_{b\in B}S_b^{ALG}}v_s'\geq \sum_{b\in B}v_b(S_b^{ALG'})-\sum_{s\in\bigcup_{b\in B}S_b^{ALG'}}v_s' \]
simply because $ALG$ maximizes the objective.
Also, due to our assumption $v_s\leq v_s'$ for all items assigned by $ALG'$, we get that the  right side is equal to $\sum_{b\in B}\hat{T}^{ALG'}_b$, yielding
\[ \sum_{b\in B}v_b(S_b^{ALG})-\sum_{s\in\bigcup_{b\in B}S_b^{ALG}}v_s'\geq  \sum_{b\in B} \hat{T}_b^{ALG'} .\]
Further, since $ALG$ only assigns items in $\hat S$, i.e., with $v_s\leq v_s'$, we can replace the left side in the same way and get
\[ \sum_{b\in B}\hat{T}_b^{ALG}\geq  \sum_{b\in B} \hat{T}_b^{ALG'} .\]
Now taking the expectations, it holds
\[ \mathbb{E}\left[\sum_{b\in B} \hat{T}_b^{ALG} \right]\geq \mathbb{E}\left[\sum_{b\in B} \hat{T}_b^{ALG'} \right]\geq \mathbb{E}\left[ \sum_{b\in B} \hat{T}^{1/2}_b \right]\]
because for each $b$, above we showed $\mathbb{E}[\hat{T}_b^{ALG'}]\geq \mathbb{E}[\hat{T}_b^{1/2}]$.
Plugging in our bound $\mathbb{E}\left[\sum_{b\in B}\hat{T}_b^{1/2}\right]\geq \tfrac 12 \mathbb{E}\left[ \sum_{b\in B}\hat{T}_b^{OPT_{max}}\right]$ implies
\[ \mathbb{E}\left[\sum_{b\in B} \hat{T}_b^{ALG} \right] \geq \mathbb{E}\left[ \frac 12 \sum_{b\in B}\hat{T}_b^{OPT_{max}} \right]\]
Finally, we conclude the proof by observing that $ \sum_{b\in B}T^{ALG}_b\geq \sum_{b\in B}\hat{T}_{b}^{ALG}.$ 
\end{proof}

With these lemmas at hand, we can now prove the claimed approximation guarantee.

\begin{theorem}[Adjusted VCG]
    For subadditive buyer valuations, the social welfare output by the adjusted-objective VCG mechanism yields a 2-approximation to the optimal welfare, in expectation. Moreover, the mechanism is IR, DSIC, and BB.
\end{theorem}
\begin{proof}
Note that we are interested in the welfare of the allocation, not its adjusted objective. It holds
\begin{align*}
    \mathbb{E}\left[ ALG\right] &= \sum_{b\in B} \mathbb{E}\left[T^{ALG}_b \right] +  \sum_{s\in S} \mathbb{E}\left[v_s\right] \ge \frac 12 \mathbb{E}\left[\sum_{b\in B}\hat{T}^{OPT_{max}}_b \right] +  \sum_{s\in S} \mathbb{E}\left[v_s \right]\\
    &\ge \frac 12  \sum_{b\in B} \mathbb{E}\left[T^{OPT}_b \right]-\frac 12 \sum_{s\in S} \mathbb{E}\left[ \left(v'_s-v_s\right)_+ \right]+\sum_{s\in S}\mathbb{E}\left[ v_s\right] \\
    &\geq \frac 12  \sum_{b\in B}\mathbb{E}\left[T^{OPT}_b \right]+\frac 12  \sum_{s\in S}\mathbb{E}\left[v_s \right] =\frac 12 \mathbb{E}\left[ OPT \right].
\end{align*}
Here, for the first inequality we used \cref{AdjVCGMax}, and for the second \cref{OptMax}. Finally, the last inequality is true because for each seller $s$, it holds that $\mathbb{E}\left[ \left(v'_s-v_s\right)_+\right]\leq \mathbb{E}\left[ v_s\right].$ In fact, $(v'_s-v_s)_+\le v_s'$ and $v_s$ and $v'_s$ share the same distribution.
\end{proof}

We conclude this section by showing that the result in \cref{thm:2-subadditive} is best-possible in the sense that no deterministic mechanism with one seller sample as sole prior information can achieve a better approximation.
\begin{theorem}    
    Every deterministic IC, IR, and BB mechanism for bilateral trade that receives as sole prior information a single sample from the seller's distribution has approximation ratio at least $2$. 
\end{theorem}

\begin{proof}
    Because the mechanism is DSIC, the reported bid of the seller cannot affect \emph{how much} she is paid for a given outcome. 
    Her price can however depend on the sample and the buyer's report. By DSIC, WBB and IR we also know that the price the buyer needs to pay whenever she gets the item does not depend on her report and needs to be higher than the seller price but below her own valuation. Using this we construct different cases where the optimum is to always trade, but show that the mechanism cannot do so in all of them. First, we assume that there exist a sample $\hat{s}$ and a buyer report $b \ge 2 \hat{s} + \epsilon$ such that the price paid to the seller for a trade is greater than $b \ge 2 \hat{s} + \epsilon$. This suggests that the mechanism overcharges: indeed, by taking advantage of DSIC, WBB and IR we can show that a buyer with value $2 \hat{s} + \epsilon/2$ would miss the trade. This implies that the mechanism doesn't charge more than twice the seller sample when a trade \emph{has} to be made to guarantee the $2-$approximation. We can therefore fix the buyer value to be very high and construct a probability distribution for the seller ranging over many small, exponentially decreasing values. Since the mechanism cannot increase the price too much (given the sample observed), with probability approaching 0.5 the seller will reject the trade.
    
    Let's now formalize the above argument. Let $v_b$ and $v_s$ denote the random variables corresponding to the seller and the buyer valuations and let $p_s(s, b)$ be the price paid to the seller in the event of a trade when the mechanism receives sample $s \sim F_s$ and report $b$ from the buyer. We distinguish two cases:

    \paragraph{\bf Case 1:} $\exists\;\hat{s}, b, \epsilon > 0 \text{ such that } p_s(\hat{s}, b) \ge 2\cdot \hat{s} + \epsilon$ and $b \ge2\cdot \hat{s} + \epsilon$. 
    In this case consider a seller whose value is always exactly $\hat{s}$ and a buyer whose value is exactly $2\cdot \hat{s} + \epsilon/2$. Clearly, a trade should always happen and the optimal welfare is $2\cdot \hat{s} + \epsilon/2$. Assuming that trade happens: by weak budget balance and individual rationality, the price $p_b(\hat{s},2\cdot \hat{s} + \epsilon/2)$ paid by the buyer must satisfy:
    $$
     \hat{s} \le  p_s(\hat{s}, 2\cdot \hat{s} + \epsilon/2) \le p_b(\hat{s},2\cdot \hat{s} + \epsilon/2) \le 2\cdot \hat{s} + \epsilon/2.
    $$
    However, by DSIC and given that a trade also needs to happen for $v_b = b$, we have that:
    $$
    p_b(\hat{s},b) \ge 2\cdot \hat{s} + \epsilon > 2\cdot \hat{s} + \epsilon/2 \ge p_b(\hat{s},2\cdot \hat{s} + \epsilon/2).
    $$
    This violates the DSIC condition, as the buyer can obtain a better price by misreporting, leading to a contradiction. Therefore, one of the two trades does not occur and the approximation ratio is at least two.

    \paragraph{\bf Case 2:} 
    $\forall \ \hat{s}, b \quad p_s(\hat{s},b ) \le 2\cdot \hat{s}$ or $b \le 2\cdot \hat{s}$. 
    Consider a seller whose valuation $v_s$ takes a uniformly random value in $\{1/3,1/3^2,\ldots,1/3^k\}$
    and a buyer with fixed value $v_b = 1$. Clearly, the optimal welfare is 1. Notice that the probability of a trade happening is:
    \begin{align*}
        \prob{p_s(\hat{s},1) \ge \seller}
        &= \sum_{i=1}^k \prob{p\left(\frac{1}{3^i},1\right)
            \ge \seller} \cdot \prob{\hat{s} = \frac{1}{3^i}}\\
        &=\sum_{i=1}^k \prob{p\left(\frac{1}{3^i},1\right) \ge \seller} \cdot \frac{1}{k}\\
        &\leq \sum_{i=1}^k \frac{k-i}{k} \cdot 
            \frac{1}{k}=\frac{k-1}{2\cdot k}
    \end{align*}
    where we used that $p(1/3^i,1) < 1/2^{i-1}$ therefore the price posted does not reach the next possible valuation. Clearly, this converges to a $2$-approximation as $k$ grows to infinity.
\end{proof}

\section{Computationally Efficient Combinatorial Mechanisms}
\label{section:combinatorial}

In this section we study the \textsc{Surplus Mechanism}, which achieves the property described in \cref{thm:black-box-1}. In \cref{sec:2xos} we present the algorithm and in \cref{sec:2xos-analysis} we show how it achieves the claimed properties. Finally, in \cref{sec:computation} we discuss computational aspects.

\subsection{The Surplus Mechanism}
\label{sec:2xos}

The basic idea behind our \textsc{Surplus Mechanism} (\cref{Algo2XOS}) is to run the given truthful one-sided mechanism $M_\alpha$ on \emph{discounted buyer valuations} and on \emph{a subset of the sellers}.
Note that the problem can be viewed as finding a hypermatching in a bipartite hypergraph $G=(\mathcal{B}\cup \mathcal{S}, E, v_{\mathcal B})$ with hyperedge set $E$ defined as all tuples $(b,S)$ s.t. $b\in \mathcal{B},$ $S\subseteq \mathcal{S}$. 

First, given valuations $v_{\mathcal{S}}$ and samples $v'_\mathcal{S}$ for each seller, we determine a subset of the sellers $\hat{S}$ as follows. For each $s \in \mathcal{S}$ we put $s$ in $\hat{S}$ if $v_{s} \leq v'_{s}$. Otherwise, we will drop $s$ from our considerations.
Next we determine discounted valuations.  For a given buyer $b$ and a given set of sellers $S \subseteq \mathcal{S}$ let $a_{b,S}$ denote the additive supporting function of buyer $b$ for set $S$. We define the discounted valuation that buyer $b$ has for the set of sellers $S \subseteq \hat{S}$ as:
\begin{align}
\label{eq:adjusted}
     \hat{v}_{b}(S) = \sum_{s\in \bar{S}} (a_{b,\bar{S}}(s)-v'_s)
     ,\text{ where }\bar{S}= \argmax_{S^\star\subseteq S\cap\hat{S}}\left\{v_{b}(S^\star)-\sum_{s\in S^\star}v'_s\right\}.
\end{align}
Adjusting the valuations like this retains the XOS property of the original valuations, as shown in the following Proposition.
\begin{proposition}
    The adjusted valuations defined in \cref{eq:adjusted} are XOS.
\end{proposition}
\begin{proof}
The definition of the adjusted valuations $\hat v_b$ captures the fact that only items in $\hat{S}$ are available for trade, and whenever a buyer is assigned some item $s$, she will be required to pay an additional $v'_s$ later.
Therefore, her valuation for $S$ is no more than her valuation for the best subset of it in $\hat{S}$, minus the sum of the according $v'_s$. This function is XOS because it can be described as the maximum over the following XOS-support: for all $a$ from the support of the original $v_b$ and $s\in \mathcal{S}$, define
$\hat a(s)= (a(s)-v'_s)_+ \text{ if }s\in \hat{S}, \text{ and }0 \text{ otherwise.}$ 

For the $\bar v$, defined as 
$$
\bar v_b(S)= \left(\sum_{s\in S\cap \hat S}\left( a_{b, S}(s)-v'_s\right)\right)_+,
$$
we define the XOS support
$
\bar a(s)= (a(s)-v'_s)$ if $s\in \hat{S}$, and $0$ otherwise,
and add an additonal function $a_0$ which is $0$ for any $s\in S$.
\end{proof}

Analogously, one can prove that the same holds for the gross substitutes class:
after the adjustment, GS valuations remain GS.

Given these {\em closure} properties, we run the one-sided mechanism $M_{\alpha}$ on the resulting hypergraph $\hat G = (\mathcal{B} \cup \hat{\mathcal{S}},\hat{E}, \hat v_{\mathcal B})$ consisting of all buyers, only the sellers in $\hat S$, and hyperedge valuations $\hat{v}_{\mathcal{B}}$. This will lead to an allocation $S_1, \dots, S_n$ and payments $p_{b}^{M_{\alpha}}(S_i)$ for each $b_i \in \mathcal{B}$.

Afterwards, we assign sets $S_1, \dots, S_n$ to buyer $b_1, \dots, b_n$ increasing buyer $b_i$'s payment relative to the payment in the one-sided mechanism by the sum of the samples $v'_{s}$ for $s \in S_i$ and pay each seller $s \in \hat{S}$ whose item has been sold the respective sample $v'_{s}$.

The construction given in our mechanism is stated in the value-oracle model, and direct computation of the adjusted valuations would be inefficient. We provide a discussion on how to implement the mechanism efficiently in \cref{sec:computation}.
For purposes of our analysis, we assume that the one-sided mechanism $M_{\alpha}$ always assigns each buyer an inclusion-minimal set of items giving the according buyer at least the same utility (this can, e.g., be ensured by employing a simple type of tie-breaking which favors small sets over larger ones).

\begin{algorithm}[ht]
\caption{\textsc{Surplus Mechanism}}
\label{Algo2XOS}
Set $\hat S = \emptyset $, $\hat E = \emptyset$\\
\For{\emph{all} $s\in S$}{
    Propose to $s$ a price of $v_s'$\\
    \If{$s$ \emph{accepts}}{
        Set $\hat{S}= \hat{S} \cup\{s\}$\\
    }
}
\For{\emph{all} $(b,S) \in B\times 2^{\hat{\mathcal{S}}}$}{
    $\hat{v}_b(S)= \sum_{s\in \bar S} (a_{b,\bar S}(s)-v'_s),\text{ where }\bar S= \argmax_{S^{\star}\subseteq S\cap\hat{S}}\left\{v_b(S^{\star})-\sum_{s\in S^{\star}}v'_s\right\}$\\
    $\hat E = \hat E \cup \{(b,S)\}$
}
 Let $A$ be the assignment on $\hat G$ induced by running $M_{\alpha}$ on the hypergraph $\hat G=(\mathcal{B} \cup \hat{\mathcal{S}}, \hat E, \hat v_{\mathcal B})$, with hyperedge weights $\hat v_{\mathcal{B}}$, presenting buyers to $M_{\alpha}$ according to its input requirements (e.g., offline or in random order)\\
\For{\emph{all} $(b,S)\in A$}{
    $b$ pays price $p_{(b,S)}= p_b^{M_{\alpha}}(S)+\sum_{s\in S}v'_s$, where $p_b^{M_{\alpha}}(S)$ is the price charged to $b$ by $M_{\alpha}$\\
    $b$ gets assigned the items in $S$\\
    \For{\emph{\textbf{each}} $s \in S$}{
        $s$ receives a payment of $v'_s$\\
    }
}
\end{algorithm}

\subsection{Analysis of \textsc{Surplus Mechanism}}
\label{sec:2xos-analysis}
We start by restating \cref{thm:black-box-1}, then we prove it by arguing in \cref{lem:surplus_thruthful} that the two-sided mechanisms inherit the IR and DSIC properties from the one-sided mechanism and finally showing the welfare approximation in \cref{lem:surplus_approx}.
\begin{theorem}[Black Box I]
Denote by $\alpha$ the approximation guarantee of any one-sided IR, DSIC offline/online mechanism for maximizing social welfare for XOS valuations. We give a two-sided mechanism for XOS buyers and unit-supply sellers that is IR, DSIC, BB, uses a single sample from each seller and provides a $\max\{2\alpha,3 \}$-approximation to the optimal social welfare. The two-sided mechanism inherits the offline/online properties of the one-sided mechanism on the buyer side and is offline on the seller side.
\end{theorem}

We start by establishing the individual rationality, truthfulness, and budget balance properties of the two-sided mechanism claimed in \cref{thm:black-box-1}.

\begin{lemma}
\label{lem:surplus_thruthful}
  Given that an IR and DSIC one-sided mechanism $M_\alpha$ is used, the two-sided \textsc{Surplus Mechanism} is DSIC, IR, and WBB.
\end{lemma}
\begin{proof}
  Let's fix any realization of the valuations. We start by showing truthfulness. Fix a seller $s\in \mathcal{S}$: the only interaction $s$ has with the algorithm is by accepting or rejecting the posted price $v_s'$, which is independent of $v_s$, in exchange for her item, which is clearly truthful and individually rational.
  
  Fixing a buyer $b$: the algorithm will ask about her valuation and modify it to $$\hat v_b(S)= \sum_{s\in \bar{S}} (a_{b,\bar S}(s)-v'_s),\text{ where }\bar S= \argmax_{S^\star\subseteq S\cap\hat{S}}\left\{v_b(S^\star)-\sum_{s\in S^\star}v'_s\right\},$$reflecting that \emph{any} item $s$ will cost at least $v'_s$. This is done because $M_{\alpha}$ only involves items, not sellers, so the extra $v'_s$ is charged afterwards. 
\begin{align*}
    &\argmax
    \left\{v_b(S') -
    \sum_{s\in S'}v'_s - p_b^{M_{\alpha}}(S') \fwhs v_b(S') \geq \sum_{s\in S'}v_s' + p_b^{M_{\alpha}} (S')\right\}\\
    &=\argmax
    \{\hat{v_b}(S') - p_b^{M_{\alpha}}(S') | \hat{v}_b(S') \ge p_b^{M_{\alpha}}(S')\}  
\end{align*}
Here, the $\argmax$s are restricted to available $S'$. The left hand side reflects that $M_{\alpha}$ is truthful and the right hand side is exactly what buyer $b$ is trying to maximize. 
Given that no trade is generated if for all sets $v_b(S') < \sum_{s\in S'}v_s'$, the mechanism is also individually rational.
For any trade, the seller's price is $v'_s$ and the buyer's $v'_s + p_b^{M_{\alpha}}(s) \geq v'_s$, so the mechanism is 
budget balanced.
\end{proof}
We now give the proof of the approximation ratio in \cref{thm:black-box-1}.

\begin{lemma}
\label{lem:surplus_approx}
  Given that an $\alpha$-approximate one-sided mechanism $M_\alpha$ is used, the two-sided \textsc{Surplus Mechanism} is $\max\{3,2\alpha\}$-approximate.
\end{lemma}

\begin{proof}
Recall the adjusted buyer valuations in the graph $\hat G$ defined in \Cref{eq:adjusted}.
Fix a pair $(b,S)$ of buyer and set of items and a realization $v_{\mathcal{B}}$ of buyers' valuations. Then, in expectation over the sellers' valuations and samples, it holds that:
\begin{equation}
\label[ineq]{eq:2vs}
        \mathbb{E}[v'_s| s \in \hat S] = \mathbb{E}[\max\{v_s,\,v'_s\}]\leq \mathbb{E}[2v_s],
\end{equation}
  where we used that  $s\in \hat S$ if and only if $v_s\leq v'_s$. Note that for every pair $(b, S')$, and every realization of buyer and seller valuations and samples implies an according maximizing set in the definition of $\hat v_b(S')$, which we will denote as $\bar S(S')$. With this, we show
\begin{align*}
    \mathbb{E}\left[\hat v_b(S)\right]& =\mathbb{E}\left[ \sum_{s\in \bar S(S)} \left( a_{b,\bar S(S)}(s)-v'_s\right) \right]\geq \mathbb{E}\left[ \sum_{s\in S\cap \hat{S}} \left( a_{b,S}(s)-v'_s\right)_+ \right]
    \\
    &=  \sum_{s\in S\cap \hat{S} } \mathbb{E}\left[\left( a_{b,S}(s)-v'_s\right)_+ \right] = \sum_{s\in S}\mathbb{E}\left[ \left(a_{b, S}(s)-v'_s\right)_+ | s \in \hat S\right]\mathbb{P}\left( s\in \hat{S} \right)
    \\
    &\geq \frac 12 \sum_{s\in S}\mathbb{E}\left[a_{b,S}(s)-v'_s|s\in\hat{S}\right] \geq \frac 12\sum_{s\in S}\mathbb{E}\left[ a_{b,S}(s)-2v_s  \right]\\
    &= \frac 12 \, \sum_{s\in S}\mathbb{E}\left[ a_{b,S}(s)\right]-\sum_{s\in S}\mathbb{E}\left[v_s\right]= \frac 12 v_b(S) -\sum_{s\in S}\mathbb{E}\left[v_s\right].
\end{align*}
The first inequality follows from the fact that since $\bar S$ maximizes $v_b(S^{\star})-\sum_{s\in S^{\star}}v'_s$, always choosing those $s\in S\cap \hat{S}$ for which $a_{b,S}(s)-v'_s\geq 0$ can only perform worse. Basic transformations and the fact that every $s\in\mathcal{S}$ is included in $\hat S$ with probability at least $\tfrac 12$ result in the second inequality, at which point we simply plug in \cref{eq:2vs}.  

For the social welfare induced by a maximum-welfare assignment $OPT$ in graph $G$, we have 
$$
SW_{OPT} =\sum_{(b,S')\in OPT} v_b(S') + \sum_{s\in S\setminus OPT} v_s.
$$
Note that in our mechanism, each buyer will only be assigned a valuation-maximizing, inclusion-minimal set $\bar S$ from the definition of $\hat v_b$.
Denote by $OPT_{1}$ an optimal assignment (hypermatching) in the original graph $G$, i.e., an optimal solution to the one-sided problem with original valuations $v_b$.
Then, for the social welfare induced by our assignment $A$ over $\hat G$, it holds 
\begin{align*}
    \mathbb{E}\left[SW_A\right] & = \mathbb{E}\left[ \sum_{(b,\bar S)\in A} v_b(\bar S)+ \sum_{s\in \mathcal{S}\setminus A}v_s \right] \\
    &\geq \mathbb{E}\left[ \sum_{(b,\bar S)\in A} \hat v_b(\bar S) + \sum_{s\in \mathcal{S}} v_s \right]\\
    &\geq \mathbb{E}\left[ \sum_{(b,S)\in OPT_{1}}\frac 1{\alpha}\hat v_b(S)+\sum_{s\in \mathcal{S}}v_s \right]\\
    &\geq \mathbb{E}\left[ \sum_{(b,S)\in OPT_{1}} \frac 1{\alpha}\left( \frac 12 v_b(S)-\sum_{s\in S}v_s \right) +\sum_{s\in \mathcal{S}}v_s \right]\\
    & = \mathbb{E}\left[ \sum_{(b,S)\in OPT_{1}} \frac{1}{2\alpha}v_b(S)-\sum_{s\in OPT_{1}}\frac 1{\alpha}v_s +\sum_{s\in \mathcal{S}}v_s \right]\\
    & \geq \mathbb{E}\left[ \sum_{(b,S)\in OPT_{1}} \frac{1}{2\alpha} v_b(S) + \left(1-\frac{1}{\alpha}\right) \sum_{s\in \mathcal{S}}v_s  \right]\\
    &\geq \mathbb{E}\left[ \sum_{(b,S)\in OPT} \frac{1}{2\alpha}v_b(S)+ \left(1-\frac{1}{\alpha}\right)\sum_{s\in \mathcal{S}}v_s \right].
\end{align*}
Here, the first inequality stems from the fact that the $\hat v_b$ are already discounted by $\sum_{s\in \bar S}v'_s\geq v_s$. 
Since $M_{\alpha}$ is an $\alpha$-approximation to the optimum in $\hat G$, also the second inequality also holds true.
Finally, we plug in the lower bound to $\hat v_b$ proven above. The expectation here is taken over all realizations of buyer valuations, seller valuations, and samples of seller valuations. 
For $\alpha \ge \frac 32$, from the previous inequality, one gets 
$$
    \mathbb{E}\left[SW_A\right] \geq \frac{1}{2\alpha} \mathbb{E}\left[ OPT \right].
$$
For any other $\alpha \in [1,\frac 32)$, we observe that $M_{\alpha}$ is also a $\frac 32$ approximation and then the same calculation gives that $
    \mathbb{E}\left[SW_A\right] \geq \frac 13 \mathbb{E}\left[ OPT \right].
$
So, all in all, we get a $\max\{2\alpha,3\}-$approximation, as claimed. 
\end{proof}

\subsection{Computational Aspects}
\label{sec:computation}

As already mentioned earlier, the adjusted valuations $\hat v_{\mathcal{B}}$ are stated with an $\argmax$ over subsets of $S$, which---depending on the computational model---may not be an efficient operation.

An alternative which works for GS, is, for all $S\subseteq \mathcal{S}$, to define $\bar v_b(S)= \left(\sum_{s\in S\cap \hat{S}} \left(a_{b, S}(s)-v'_s\right)\right)_+$. As before it can be shown that if the original valuations are GS, then the modified  valuations are GS as well. A difference between $\bar{v}_{\mathcal{B}}$ and $\hat{v}_{\mathcal{B}}$ is that the former need not be monotone, but monotonicity is not required by poly-time algorithms for GS valuations (see, e.g., \cite{PaesLeme17}). Moreover, the approximation guarantee of the one-sided mechanism run on $\bar{v}_{\mathcal{B}}$ also applies if the resulting assignment is evaluated with the original adjusted valuations $\hat{v}_{\mathcal{B}}$ and the benchmark is the optimal allocation under the original adjusted valuations $\hat{v}_{\mathcal{B}}$, which is the only property of the one-sided mechanism that we used in our proof above. This is clarified in the following proposition.

\begin{proposition}
\label{prop:computational}
    Any inclusion-minimal assignment $ALG$ made by the algorithm run with adjusted valuations $\bar v$ instead of $\hat v$ provides the same approximation to the optimum welfare.
\end{proposition}

\begin{proof}
Note that the $\bar v_b$ are in the class XOS, so $M_\alpha$'s approximation guarantee holds up.
Also, for any pair $(b,S)$ with $b\in B$, $S\subseteq \hat S$:
$
\hat v_b(S) \geq \bar v_b(S)
$
, since $\hat v_b(S)$ results from picking an optimal subset from $S$.

Let $S^{opt}$ be an inclusion-minimal, utility-preserving set of items assigned to buyer $b$, which is a subset of $b$'s assigned bundle $S^b$ in an optimal hypermatching $OPT(\hat G)$ of graph $\hat G$ as defined in our algorithm, i.e.,
\begin{align*}
\bar v_b(S^{opt}) &= \sum_{s\in S^{opt}} (a_{b,S^{opt}}(s)-v'_s),
\text{ where } S^{opt}= \argmax_{S^\star\subseteq S^{b}\cap \hat{S}}\left\{v_b(S^\star)-\sum_{s\in S^\star}v'_s\right\}.
\end{align*}
By definition of $\bar v(S^{opt})$, it holds that the weight of an optimal assignment is the same for both $\bar v$ and $\hat v$. 
Combining these facts, we get for the weight of the hypermatching $A_{\bar v}$ returned by the algorithm when using the valuations $\bar v_{\mathcal{B}}$:
\[
\sum_{(b,S)\in A_{\bar v}} \hat v_b(S) \geq \sum_{(b,S)\in A_{\bar v}} \bar v_b(S)\geq c\cdot OPT(\bar G)=c\cdot OPT(\hat G)
\]
and this is what we needed in the proof of the approximation ratio.
\end{proof}

In addition, note that while our mechanism is stated for the case of value queries, it can be formulated also for demand queries---as required by poly-time mechanisms for XOS valuations such as \cite{Dobzinski16,AssadiS19,AssadiEtAl21}---by adjusting the prices proposed instead of the buyer valuations. Here, it is simply necessary to increase the $M_{\alpha}$-prices of any $S\subseteq \mathcal{S}\cap \hat S$ by $\sum_{s\in S}v'_s$. This does not reflect the \emph{capping} of our adjusted valuations to a minimum contribution of $0$ for each seller, but---having the same effect---buyers will never demand the according items.

Finally, observe that if we are allowed to use demand queries, then we can also efficiently implement value queries (as required by some demand-oracle algorithms) to the adjusted valuations $\hat{v}_\mathcal{B}$: first issue a demand query to find the set $S^\star$ in the argmax, and then issue a value query to obtain the value $v_b(S^\star)$ of the corresponding set.

\section{Computationally Efficient Double Auctions}
\label{section:double}

We now move on to less general settings with unit-demand buyers and unit-supply sellers with \emph{identical items}. We present a general technique for turning truthful one-sided mechanisms into truthful two-sided mechanisms while enforcing the strong budget balance property (in contrast to just weak budget balance from the previous section).  

The one-sided mechanisms we consider are for so-called binary single-parameter problems. In such a problem, an agent $i$ can either win or lose, and has a value $v_i$ for winning. The set of agents that can simultaneously win is given by a set system $\mathcal{I}$. The social welfare of a feasible set $X \in \mathcal{I}$ is simply the sum $\sum_{i \in X} v_i$ of the winning agents' valuations.

Given two set systems $\mathcal{I}$ and $\mathcal{I}'$ on a ground set $U$, we define its intersection to be the set system that contains all sets $X \subseteq U$ such that $X \in \mathcal{I}$ and $X \in \mathcal{I}'$.

\begin{theorem}[Black Box II]
Denote by $\alpha$ the approximation guarantee of any one-sided IR, DSIC offline/online mechanisms for the intersection of a downward closed set system $\mathcal{I}$ with a uniform matroid that may or may not use information on the priors.
We give a two-sided mechanism for unit-demand buyers, unit-supply sellers with identical goods, and constraints $\mathcal{I}$ on the buyers, that is IR, DSIC,  SBB, has the same information requirements on the buyer side as the one-sided mechanism and uses a single sample from each seller, and yields a 
$(1 + 1/(2-\sqrt{3})\cdot \alpha) \approx (1+3.73 \cdot \alpha)$
approximation to the optimal social welfare. The mechanism inherits the online/offline properties of the one-sided mechanism on the buyer side and it is online random order on the seller side.
\end{theorem}

Instead of directly proving \cref{thm:black-box-2}, we present the algorithm for the following special case, 
and then argue how to generalize it. 

\begin{theorem}\label{thm:two-sided-rehearsal}
Let $k = \min\{n,m\}$. There is an IR, DSIC, SBB, $1+\frac{1}{2-\sqrt{3}}(1+O(1/\sqrt{k}))$-approximate single-sample mechanism for unconstrained double auctions that approaches the buyers in online fixed order and the sellers in online random order.
\end{theorem}

We give the explicit mechanism \textsc{Reserve Rehearsal} for \cref{thm:two-sided-rehearsal} in \cref{alg:two-sided-rehearsal}.
\textsc{Reserve Rehearsal} runs the one-sided \textsc{Rehearsal} algorithm  \cite{AzarKW14} to select the top $k = \min\{n,m\}$ buyers. \citet{AzarKW14} show that the combined value of the buyers that beat their price provides, in expectation, a $1+O(\frac{1}{\sqrt{k}})$ approximation to the expected value of the $k$ highest buyers.

Our twist to this mechanism is that we pair buyers $b$ that would be selected by the \textsc{Rehearsal} algorithm with a random seller $s$, offering them to trade at a price that is the max of the respective buyer's $p_{min}$ and the respective seller's sample $v'_{s}$. This adds the needed component to take into account the valuations on the seller side of the market, and will serve as an insurance that any \emph{good} trade we propose has a good chance of actually happening.

\begin{algorithm}[ht]
\DontPrintSemicolon
\caption{\textsc{Reserve Rehearsal}}
Let $P$ be the set of the $k-2\sqrt k$ largest buyer samples, together with $2\sqrt k$ copies of the $(k-2\sqrt k)^{th}$-largest buyer sample\\
Let $p_{min}$ be the smallest element of $P$\\
Fix any order on the buyers (or assume buyers arrive online)\\
\For{\emph{each} $b_i$, \emph{in this order}}{
    \If{$v_{b_i} > p_{min}$}{
    Delete from $P$ the highest value $p\in P$ such that $v_{b_i}>p$\\
    Pick uniformly at random $s \in \mathcal{S}$, delete $s$ from $\mathcal{S}$ (or assume sellers arrive online)\\
    Propose a trade to $(b_i,s)$ for the price of $\max\{v'_{s}, p_{min}\}$\\
    \If {$b_i$ \emph{and} $s$ \emph{both agree}}{
        Make the trade at this price
    }
    }
}
\label{alg:two-sided-rehearsal}
\end{algorithm}

For the more general result in \cref{thm:black-box-2} we run the given one-sided mechanism on the intersection of the given feasibility constraint $(\mathcal{B},\mathcal{I}_\mathcal{B})$ with a $m$-uniform matroid. This gives a set of tentative buyers $B'$ along with tentative buyer prices $p_{\mathcal{B}}$. We can then randomly match the tentative buyers to sellers, offering buyer-seller pairs $(b,s)$ to trade at price $\max\{p_{b},v'_{s}\}$.

We begin by showing that our two-sided mechanism inherits IR and DSIC from its one-sided counterpart, and that it is strongly budget balanced.

\begin{lemma}\label{lem:rehearsal-truthful}
    The \textsc{Reserve Rehearsal} mechanism is IR, DSIC, and SBB.
\end{lemma}
\begin{proof}
    In \textsc{Reserve Rehearsal} no money is ever received by the mechanism itself. The only exchange of money happens between buyer-seller pairs $(b_i,s)$ that also exchange an item.
The mechanism is therefore strongly budget balanced. 

It is also clear that the mechanism is individually rational as buyers and sellers would only accept trades at prices that are lower resp.~higher than their respective valuations.

Furthermore, the mechanism is truthful for agents on both sides of the market.
Each buyer is presented with a trading opportunity once; and the price depends only on the samples and the valuations of previously considered buyers. She can only accept or reject, but never influence it---and will therefore not profit from reporting a lower or higher value.
Sellers, on the other hand, are also guaranteed to be considered only once. They, too, have no means of influencing the price they are presented with, and can only accept or reject.
\end{proof}

It remains to show the claimed approximation guarantee. Just as in the case of bilateral trade, we will do the bulk of the work for fixed buyer valuations. 

In what follows, we use $M = (x,p)$ to refer to the one-sided version of \textsc{Rehearsal}, and $M' = (x',p')$ to refer to the two-sided version.
We use $B'$ to denote the set of tentative buyers chosen by $M$, we use $B'_{+}$ to denote the 
set of buyers that end up with an item in $M'$, and we use $S_{+}$ to denote the set of sellers that keep their item in $M'$. 
The expected social welfare achieved by $M'$ is: 
$$   	\expectation[\textsf{SW}(x'(v_\mathcal{B},v_\mathcal{S}))] = \expectation\Big[\sum_ {b \in B'_+} v_{b}\ + \sum_{s \in S_+} v_{s}\Big]. 
$$
The key bit in our proof is the following lemma, which relates the performance of the one-sided mechanism to that of the two-sided mechanism.

\begin{lemma}\label{lem:key-lemma-identical}
Let $B'$ denote the set of tentative buyers chosen by the one sided mechanism $M$, let $B'_+$ denote the set of buyers that trade in the two-sided mechanism $M'$, and let $S_+$ denote the set of sellers that keep their item in the two-sided mechanism $M'$. Then,
\begin{align*}
\expectation\Big[\sum_ {b \in B'_+} v_{b}\ + \sum_{s \in S_+} v_{s}\Big] \geq (2-\sqrt{3}) \cdot \expectation\Big[\sum_{b \in B'} v_{b}\Big].
\end{align*}
\end{lemma}

\begin{proof}
In order to prove the lemma we will show that for any fixed buyer valuations $v_\mathcal{B}$,  buyer samples $v'_{\mathcal{B}}$, and corresponding set of tentative buyers $B'$, in expectation over the seller valuations $v_\mathcal{S}$, the seller samples $v'_{\mathcal{S}}$, and the 
randomness in the pairing of buyers and sellers,
\[
    \expectation\Big[\sum_{b \in B'_{+}} v_{b} + \sum_{s 
    \in S_{+}} v_{s} \;\mid\; B', v_{\mathcal{B}},v'_{\mathcal{B}} \Big] \geq (2-\sqrt{3}) \cdot \sum_{b \in B'} v_{b}.
\]
        
The actual claim then follows by taking expectation over buyer valuations $v_\mathcal{B}$, buyer samples $v'_{\mathcal{B}}$, and the corresponding set of tentative buyers $B'$.

To this end, let's fix any buyer $b \in B'$ with her value $v_{b}$ and tentative payment $p_{b}$. 
Let's denote by $s$ the random seller associated to $b$, and by $v_s$ and $v'_s$ two independent samples from that seller's value distribution.
Note that from buyer $b$'s perspective seller $s$ is just a uniform random seller from $\mathcal{S}$.
The contribution of the couple $(b,s)$ to the social welfare is $v_{b}$ if there is a trade and $v_s$ otherwise, and there is a sale when 
$v_{b} \geq \max\{p_{b},v'_s\} \geq v_s$.
        
To analyze this contribution, we fix some constant $t\geq 0$ which will be set later and we denote with $F(x)=\prob{v_s\le x}$ the probability that a sample from $F_s$ of a seller $s$ chosen uniformly at random from $\mathcal{S}$ is at most $x$.
For any $t$ we can do a case distinction based on whether $v_b \geq t$ or $v_b < t$.
If $v_{b} \geq t,$ then the probability to have a sale is at least
\begin{align*}
    &\prob{\{v_s' \leq v_{b}\} \cap \{ v_s \leq \max\{p_{b} , v'_s\}\} \;|\; t \leq v_{b}} \\
    &\quad \geq 
    \prob{\{v_s' \leq t\} \cap \{ v_s \leq v_s'\}}\\
    &\quad \ge \prob{\{v_s' \leq t\} \cap \{ v_s \leq v_s'\} \cap \{v_s \le t\}}\\
    &\quad\ge 
    \prob{v_s \leq v_s' \;|\; \{v_s' \leq t\} \cap \{ v_s \le t\}}\prob{\{v_s \le t\} \cap \{ v_s' \le t\}}\\
    &\quad \geq \tfrac 12 F(t)^2.
\end{align*}
So the expected contribution is at least 
$\frac{1}{2}F(t)^2 v_{b}$. For the case $v_b < t$ we use the expected value of the seller as a lower bound to obtain that the expected contribution is at least
\begin{align*}
    \CondE{}{v_s}{t > v_b} 
    &= \E{}{v_s} 
    = \E{}{v_s'} 
    \geq 
    \CondE{}{v_s'}{v_s' \geq v_b} \cdot \prob{v_s' \geq v_b}  \\
    &\geq(1-F(v_b)) v_b\geq \left(1-F\left(\tfrac{t+v_b}{2}\right)\right) v_b.
\end{align*}
So the expected contribution of buyer $b$ and her random partner $s$ to the social welfare is at least $v_b \cdot LB(t)$, where
$$
LB(t) \geq
\begin{cases}
    \tfrac 12 F(t)^2 \quad & \text{ if } t \leq v_b\\
    1-F\left(\frac{v_b+t}{2}\right)               &\text{ if }t> v_b
\end{cases}
$$

We now set the arbitrary parameter $t\ge 0$ to lowerbound that term. Consider $t^\star= \min\{t| F(t) \ge \sqrt 3 - 1\}$, which is  the (possibly jumping) point of function $F$ for which $F(t^\star)\ge \sqrt{3}-1 $, and $F(t')<\sqrt{3}-1$ for all $t' < t^\star$.
Using this $t^\star$ we can lower bound the expected contribution of $B_i$ and $s$ as follows:
\begin{align*}
v_b \cdot LB(t)&\geq v_b\cdot\min\left\{ \tfrac 12 F(t^\star)^2,\, 1-F\left(\tfrac{v_b+t^\star}{2}\right) \right\}\\
&\geq v_b\cdot\min\left\{ \tfrac 12 2(2-\sqrt{3}),\, 2-\sqrt{3} \right\}\\
&= v_b\cdot\min\left\{ 2-\sqrt{3},\, 2-\sqrt{3} \right\}=v_b\cdot (2-\sqrt{3}).
\end{align*}

One concise way to express the progress so far is: 
$$
	\expectation\Big[\sum_{a \in (B'_+ \; \cup \; S_+) \cap \{b,s\}}v_a \;\mid\; B',v_{\mathcal{B}},v'_{\mathcal{B}}\Big] \geq  (2-\sqrt{3}) \cdot 
    v_{b}, \ \ \ \forall b \in B'
$$
where the randomness is over the choice of the random seller $s$ and her values $v_s$ and $v_s'$. The above holds for all buyers $b$, so we can sum up for all buyers in $B'$, then use linearity of expectation and the fact that $|B'| \leq m$ to obtain
\begin{align*}
    \sum_{b \in B'} (2-\sqrt{3}) \cdot v_{b} &\le  
    \sum_{b \in B'} \expectation\Big[\sum_{a \in (B'_+ \; \cup \; S_+) \cap \{b,s\}}v_a \;\mid\; B',v_{\mathcal{B}},v'_{\mathcal{B}}\Big]\\
    &\leq \expectation\Big[\sum_{b \in B'_{+}} v_{b} + \sum_{s \in \mathcal{S}_{+}} v_{s} \;\mid\; B',v_{\mathcal{B}},v'_{\mathcal{B}}\Big],
\end{align*}
as claimed.
\end{proof}

We are now ready to prove the performance guarantee of the two-sided mechanism $M'$. This, together with \Cref{lem:rehearsal-truthful}, concludes the proof of \cref{thm:two-sided-rehearsal}.

\begin{lemma}
\label{lem:main-two-sided-rehearsal}
Fix $n$ and $m$ and let $\alpha$ denote the approximation guarantee of the \textsc{Rehearsal} algorithm.
The \textsc{Reserve Rehearsal} mechanism yields in expectation a 
$$\left(1 + \frac{1}{2-\sqrt{3}}\cdot \alpha\right) \approx (1+3.73 \cdot \alpha)$$ 
approximation to the expected optimal social welfare.
\end{lemma}
\begin{proof}
Recall that we use $B'_+$ to denote the set of buyers that actually do get an item in our two-sided mechanism $M'$, and that we use $S_+$ to denote those sellers that do not make any trade and keep their item.
With this notation the expected social welfare achieved by our two-sided mechanism $M'$ is
\begin{align*}
\expectation[\textsf{SW}(x'(v_{\mathcal{B}},v_{\mathcal{S}}))] = \expectation\Big[\sum_{b \in B'_+} v_{b} + \sum_{s \in S_+} v_{s}\Big]. 
\end{align*}
For a given set of valuations $v_\mathcal{B}$ of the buyers, denote by $OPT_k(v_\mathcal{B})$ the set of buyers with the $k$ highest values. 
We can upper bound the expected optimal social welfare by the optimal solution for the buyers plus all seller values
\begin{align} 
\expectation[\textsf{SW}(OPT(v_{\mathcal{B}},v_{\mathcal{S}}))] \leq \expectation\Big[\sum_{b \in OPT_k(v_\mathcal{B})} v_{b}\Big] + \expectation\Big[\sum_{s \in \mathcal{S}} v_{s}\Big].
\label[ineq]{eq:basic-upper-bound-on-opt}
\end{align}
Recall that the one-sided mechanism $M$ computes a set $B'$ of buyers whose accumulated expected values are at least $\frac 1{\alpha}$ times the expected one-sided optimum. Hence, for the considered buyers $B'$,
$$
\frac{1}{\alpha} \cdot \expectation\Big[\sum_{b \in OPT_k(v_\mathcal{B})} v_{b}\Big]\leq \expectation\Big[\sum_{b \in B'} v_{b} \Big]  .
$$
By combining the above inequality with our upper bound on the expected optimal social welfare in \cref{eq:basic-upper-bound-on-opt}, we obtain
\begin{align}
\expectation[\textsf{SW}(OPT(v_{\mathcal{B}},v_{\mathcal{S}}))] \leq \alpha \cdot  \expectation\Big[\sum_{b \in B'} v_{b}\Big] + \expectation\Big[\sum_{s \in \mathcal{S}} v_{s} \Big].
\label[ineq]{eq:upper-bound-on-opt}
\end{align}
First consider the second term on the right hand side of \cref{eq:upper-bound-on-opt}. In our two-sided mechanism $M'$ sellers trade only if they are matched to a buyer with higher valuation. 
Therefore, we can replace as follows:
\begin{align*}
\expectation[\textsf{SW}(OPT(v_{\mathcal{B}},v_{\mathcal{S}}))] \leq \alpha \cdot \expectation\Big[\sum_{b\in B'} v_{b}\Big] + \expectation\Big[\sum_{s\in S_+} v_{s} + \sum_{b \in B'_+} v_{b} \Big]. 
\end{align*}

Now consider the first term on the right hand side of \cref{eq:upper-bound-on-opt}. Our two-sided mechanism $M'$ does not make trades for each buyer in $B'$. 
Despite the fact that generally 
$B'\neq B'_+$, as we show in \cref{lem:key-lemma-identical}, in expectation, 
$(2-\sqrt{3}) \cdot \expect{\sum_{b \in B'} v_{b}}$ 
is a lower bound on our mechanism's social welfare. Using this we obtain
\begin{align*}
\expectation[\textsf{SW}&(OPT(v_{\mathcal{B}},v_{\mathcal{S}}))] \leq \frac{\alpha}{2-\sqrt{3}} \cdot 
\expectation\Big[\sum_{b \in B'_+} v_{b} + \sum_{s\in S_+} v_{s}\big]+ 
\expectation\Big[\sum_{s\in S_+} v_{s} + \sum_{b\in B'_+} v_{b}\Big].
\end{align*}
All in all, we get:
\begin{align*}
\expectation[\textsf{SW}(OPT(v_{\mathcal{B}},v_{\mathcal{S}}))]
\leq 
\left(1+\frac{\alpha}{2-\sqrt{3}} \right) \cdot
\expectation\Big[\sum_{b\in B'_+} v_{b} + \sum_{s\in S_+} v_{s}\Big],
\end{align*}
as claimed.
\end{proof}

We conclude by describing how to extend this argument to prove the more general result in \cref{thm:black-box-2}. 

\begin{proof}[Skectch of the proof of \cref{thm:black-box-2}]
The proof that the general mechanism described above  is IR, DSIC, and SBB is basically identical to that of \cref{lem:rehearsal-truthful}. The key is that truthfulness of the one-sided mechanism ensures that tentative buyers want an opportunity to trade, and cannot manipulate the price they face for this opportunity.
For the performance analysis we claim that \cref{lem:main-two-sided-rehearsal} applies more generally with $\alpha$ being the approximation guarantee of the one-sided mechanism. In fact, the only change to the above proof that is required for this generalization is to redefine $OPT_k$ as the optimal one-sided solution containing at most $k = \min\{n,m\}$ buyers.
\end{proof}

\section{Conclusion and Further Directions}
We have initiated the study of simple mechanisms for two-sided markets that use only a minimum amount of information from the priors, i.e., just a single sample from the sellers' distribution. While without prior information, there is generally no hope for any meaningful approximation, our mechanisms are approximately efficient up to small constant factors.

This line of research is of specific relevance for two-sided markets since efficient mechanisms with the desired requirements
are only possible in the Bayesian setting and, moreover, the optimal Bayesian mechanism is
complicated and known only for restricted cases. 
 Our results are very general, and in several cases even improve on the best known approximation guarantee with perfect knowledge of the distributions.
Our results can be extended further by identifying new efficient mechanisms for the one-sided versions of the problems. Finally, we leave open the problem of devising single sample mechanisms for the challenging problem of optimizing the {\em gain from trade} in two-sided markets.

\clearpage

\bibliographystyle{plainnat}
\bibliography{bilateral}

\begin{thebibliography}{43}
\providecommand{\natexlab}[1]{#1}
\providecommand{\url}[1]{\texttt{#1}}
\expandafter\ifx\csname urlstyle\endcsname\relax
  \providecommand{\doi}[1]{doi: #1}\else
  \providecommand{\doi}{doi: \begingroup \urlstyle{rm}\Url}\fi

\bibitem[Assadi and Singla(2019)]{AssadiS19}
S.~Assadi and S.~Singla.
\newblock Improved truthful mechanisms for combinatorial auctions with
  submodular bidders.
\newblock In \emph{Proc. 60th IEEE FOCS}, pages 233--248, 2019.
\newblock \doi{10.1109/FOCS.2019.00024}.

\bibitem[Assadi et~al.(2021)Assadi, Kesselheim, and Singla]{AssadiEtAl21}
S.~Assadi, T.~Kesselheim, and S.~Singla.
\newblock Improved truthful mechanisms for subadditive combinatorial auctions:
  Breaking the logarithmic barrier.
\newblock In \emph{Proc. 32nd ACM-SIAM SODA}, pages 653--661, 2021.
\newblock \doi{10.1137/1.9781611976465.40}.

\bibitem[Azar et~al.(2014)Azar, Kleinberg, and Weinberg]{AzarKW14}
P.~D. Azar, R.~Kleinberg, and S.~Matthew Weinberg.
\newblock Prophet inequalities with limited information.
\newblock In \emph{Proc.~25th ACM-SIAM SODA}, pages 1358--1377, 2014.
\newblock \doi{10.1137/1.9781611973402.100}.

\bibitem[Babaioff and Nisan(2004)]{BabaioffN04}
M.~Babaioff and N.~Nisan.
\newblock Concurrent auctions across the supply chain.
\newblock \emph{J.~Artif.~Intell.~Res.}, 21:\penalty0 595–629, 2004.
\newblock \doi{10.1613/jair.1316}.

\bibitem[Babaioff and Walsh(2005)]{BabaioffW05}
M.~Babaioff and W.~E. Walsh.
\newblock Incentive-compatible, budget-balanced, yet highly efficient auctions
  for supply chain formation.
\newblock \emph{Decis. Support Sys.}, 39\penalty0 (1):\penalty0 123--149, 2005.
\newblock \doi{10.1016/j.dss.2004.08.008}.

\bibitem[Babaioff et~al.(2018)Babaioff, Cai, Gonczarowski, and
  Zhao]{Babaioff0GZ18}
M.~Babaioff, Y.~Cai, Y.~A. Gonczarowski, and M.~Zhao.
\newblock The best of both worlds: Asymptotically efficient mechanisms with a
  guarantee on the expected gains-from-trade.
\newblock In \emph{Proc.~19th ACM EC}, page 373, 2018.
\newblock \doi{10.1145/3219166.3219203}.

\bibitem[Balseiro et~al.(2019)Balseiro, Mirrokni, {Paes Leme}, and
  Zuo]{BalseiroMLZ19}
S.~R. Balseiro, V.~S. Mirrokni, R.~{Paes Leme}, and S.~Zuo.
\newblock Dynamic double auctions: {T}owards first best.
\newblock In \emph{Proc.~30th ACM-SIAM SODA}, pages 157--172, 2019.
\newblock \doi{10.1137/1.9781611975482.11}.

\bibitem[Blumrosen and Dobzinski(2014)]{blumrosen2014reallocation}
L.~Blumrosen and S.~Dobzinski.
\newblock Reallocation mechanisms.
\newblock In \emph{Proc.~15th ACM EC}, page 617, 2014.
\newblock \doi{10.1145/2600057.2602843}.

\bibitem[Blumrosen and Dobzinski(2016)]{DobzinskiB16}
L.~Blumrosen and S.~Dobzinski.
\newblock (almost) efficient mechanisms for bilateral trading.
\newblock \emph{CoRR}, 2016.
\newblock URL \url{http://arxiv.org/abs/1604.04876}.

\bibitem[Blumrosen and Mizrahi(2016)]{Blumrosen2016}
L.~Blumrosen and Y.~Mizrahi.
\newblock Approximating gains-from-trade in bilateral trading.
\newblock In \emph{Proc.~12th WINE}, pages 400--413, 2016.
\newblock \doi{10.1007/978-3-662-54110-4\_28}.

\bibitem[Brustle et~al.(2017)Brustle, Cai, Wu, and
  Zhao]{DBLP:conf/sigecom/BrustleCWZ17}
J.~Brustle, Y.~Cai, F.~Wu, and M.~Zhao.
\newblock Approximating gains from trade in two-sided markets via simple
  mechanisms.
\newblock In \emph{Proc.~20th {ACM} EC}, pages 589--590, 2017.
\newblock \doi{10.1145/3033274.3085148}.

\bibitem[Clarke(1971)]{clarke71}
E.~H. Clarke.
\newblock Multipart pricing of public goods.
\newblock \emph{Public Choice}, 11\penalty0 (1):\penalty0 17--33, 1971.
\newblock \doi{https://doi.org/10.1007/BF01726210}.

\bibitem[Cole and Roughgarden(2014)]{ColeR14}
R.~Cole and T.~Roughgarden.
\newblock The sample complexity of revenue maximization.
\newblock In \emph{Proc.~46th ACM-SIGACT STOC}, pages 243--252, 2014.
\newblock \doi{10.1145/2591796.2591867}.

\bibitem[{Colini{-}Baldeschi} et~al.(2016){Colini{-}Baldeschi}, {de Keijzer},
  Leonardi, and Turchetta]{Colini-BaldeschiEtAl16}
R.~{Colini{-}Baldeschi}, B.~{de Keijzer}, S.~Leonardi, and S.~Turchetta.
\newblock Approximately efficient double auctions with strong budget balance.
\newblock In \emph{Proc.~27th ACM-SIAM SODA}, pages 1424--1443, 2016.
\newblock \doi{10.1137/1.9781611974331.ch98}.

\bibitem[{Colini{-}Baldeschi} et~al.(2017{\natexlab{a}}){Colini{-}Baldeschi},
  Goldberg, {de Keijzer}, Leonardi, Roughgarden, and
  Turchetta]{Colini-BaldeschiEtAl17}
R.~{Colini{-}Baldeschi}, P.~W. Goldberg, B.~{de Keijzer}, S.~Leonardi,
  T.~Roughgarden, and S.~Turchetta.
\newblock Approximately efficient two-sided combinatorial auctions.
\newblock In \emph{Proc.~18th ACM EC}, pages 591--608, 2017{\natexlab{a}}.
\newblock \doi{10.1145/3033274.3085128}.

\bibitem[{Colini{-}Baldeschi} et~al.(2017{\natexlab{b}}){Colini{-}Baldeschi},
  Goldberg, {de Keijzer}, Leonardi, and Turchetta]{colini2017fixed}
R.~{Colini{-}Baldeschi}, P.~W. Goldberg, B.~{de Keijzer}, S.~Leonardi, and
  S.~Turchetta.
\newblock Fixed price approximability of the optimal gain from trade.
\newblock In \emph{Proc.~13th WINE}, pages 146--160, 2017{\natexlab{b}}.
\newblock \doi{10.1007/978-3-319-71924-5\_11}.

\bibitem[Correa et~al.(2019)Correa, D{\"{u}}tting, Fischer, and
  Schewior]{CorreaDFS19}
J.~R. Correa, P.~D{\"{u}}tting, F.~A. Fischer, and K.~Schewior.
\newblock Prophet inequalities for {I.I.D.} random variables from an unknown
  distribution.
\newblock In \emph{Proc.~20th ACM EC}, pages 3--17, 2019.
\newblock \doi{10.1145/3328526.3329627}.

\bibitem[Correa et~al.(2020)Correa, Cristi, Epstein, and Soto]{CorreaCES20}
J.~R. Correa, A.~Cristi, B.~Epstein, and J.~A. Soto.
\newblock The two-sided game of googol and sample-based prophet inequalities.
\newblock In \emph{Proc.~31st ACM-SIAM SODA}, pages 2066--2081, 2020.
\newblock \doi{10.1137/1.9781611975994.127}.

\bibitem[Dhangwatnotai et~al.(2015)Dhangwatnotai, Roughgarden, and
  Yan]{DhangwatnotaiRY15}
P.~Dhangwatnotai, T.~Roughgarden, and Q.~Yan.
\newblock Revenue maximization with a single sample.
\newblock \emph{Games and Econ. Behav.}, 91:\penalty0 318--333, 2015.
\newblock \doi{10.1016/j.geb.2014.03.011}.

\bibitem[Dobzinski(2016)]{Dobzinski16}
S.~Dobzinski.
\newblock Breaking the logarithmic barrier for truthful combinatorial auctions
  with submodular bidders.
\newblock In \emph{Proc.~{48th ACM-SIGACT STOC}}, pages 940--948, 2016.
\newblock \doi{10.1145/2897518.2897569}.

\bibitem[D\"utting et~al.(2017)D\"utting, Roughgarden, and Talgam-Cohen]{DRT17}
P.~D\"utting, T.~Roughgarden, and I.~Talgam-Cohen.
\newblock Modularity and greed in double auctions.
\newblock \emph{Games Econ. Behav.}, 105:\penalty0 59--83, 2017.
\newblock \doi{10.1016/j.geb.2017.06.008}.

\bibitem[Giannakopoulos et~al.(2017)Giannakopoulos, Koutsoupias, and
  Lazos]{giannakopoulos2017online}
Y.~Giannakopoulos, E.~Koutsoupias, and P.~Lazos.
\newblock Online market intermediation.
\newblock In \emph{Proc.~44th {ICALP}}, pages 47:1--47:14, 2017.
\newblock \doi{10.4230/LIPIcs.ICALP.2017.47}.

\bibitem[Goldner et~al.(2020)Goldner, Babaioff, and
  Gonczarowski]{GoldnerEtAl20}
K.~Goldner, M.~Babaioff, and Y.~A. Gonczarowski.
\newblock Bulow-klemperer-style results for welfare maximization in two-sided
  markets.
\newblock In \emph{Proc.~31st ACM-SIAM SODA}, pages 2452--2471, 2020.
\newblock \doi{10.1137/1.9781611975994.150}.

\bibitem[Gomes and Mirrokni(2014)]{GomesM14}
R.~Gomes and V.~S. Mirrokni.
\newblock Optimal revenue-sharing double auctions with applications to ad
  exchanges.
\newblock In \emph{Proc.~23rd WWW}, pages 19--28, 2014.
\newblock \doi{10.1145/2566486.2568029}.

\bibitem[Groves(1973)]{groves73}
T.~Groves.
\newblock Incentives in teams.
\newblock \emph{Econometrica}, 41\penalty0 (4):\penalty0 617--631, 1973.
\newblock \doi{https://doi.org/10.2307/1914085}.

\bibitem[Kang and Vondr{\'a}k(2019)]{kang2018strategy}
Z.~Y. Kang and J.~Vondr{\'a}k.
\newblock Fixed-price approximations to optimal efficiency in bilateral trade.
\newblock \emph{SSRN}, 2019.
\newblock URL \url{https://ssrn.com/abstract=3460336}.

\bibitem[Kleinberg(2005)]{Kleinberg05}
R.~D. Kleinberg.
\newblock A multiple-choice secretary algorithm with applications to online
  auctions.
\newblock In \emph{Proc.~16th ACM-SIAM SODA}, pages 630--631, 2005.

\bibitem[Koutsoupias and Lazos(2018)]{koutsoupias2018online}
E.~Koutsoupias and P.~Lazos.
\newblock Online trading as a secretary problem.
\newblock In \emph{Proc.~11th {SAGT}}, pages 201--212, 2018.
\newblock \doi{10.1007/978-3-319-99660-8\_18}.

\bibitem[McAfee(1992)]{mcafee92}
R.~P. McAfee.
\newblock A dominant strategy double auction.
\newblock \emph{J. Econ. Theory}, 56\penalty0 (2):\penalty0 434--450, 1992.
\newblock \doi{https://doi.org/10.1016/0022-0531(92)90091-U}.

\bibitem[McAfee(2008)]{mcafee2008gains}
R.~P. McAfee.
\newblock The gains from trade under fixed price mechanisms.
\newblock \emph{Appl. Econ. Res. Bull.}, 1\penalty0 (1):\penalty0 1--10, 2008.

\bibitem[Morgenstern and Roughgarden(2015)]{MorgensternR15}
J.~Morgenstern and T.~Roughgarden.
\newblock On the pseudo-dimension of nearly optimal auctions.
\newblock In \emph{Proc.~NIPS}, pages 136--144, 2015.

\bibitem[Mu'alem and Nisan(2008)]{MN08}
A.~Mu'alem and N.~Nisan.
\newblock Truthful approximation mechanisms for restricted combinatorial
  auctions.
\newblock \emph{Games Econ. Behav.}, 64\penalty0 (2):\penalty0 612--631, 2008.
\newblock \doi{10.1016/j.geb.2007.12.009}.

\bibitem[Myerson and Satterthwaite(1983)]{myersonS83}
R.~B. Myerson and M.~A. Satterthwaite.
\newblock Efficient mechanisms for bilateral trading.
\newblock \emph{J.~Econ.~Theory}, 29\penalty0 (2):\penalty0 265--281, 1983.
\newblock \doi{https://doi.org/10.1016/0022-0531(83)90048-0}.

\bibitem[Niazadeh et~al.(2014)Niazadeh, Yuan, and Kleinberg]{NiazadehYK14}
R.~Niazadeh, Y.~Yuan, and R.~D. Kleinberg.
\newblock Simple and near-optimal mechanisms for market intermediation.
\newblock In \emph{Proc.~10th WINE}, pages 386--399, 2014.
\newblock \doi{10.1007/978-3-319-13129-0\_31}.

\bibitem[Nisan and Segal(2006)]{NisanS06}
N.~Nisan and I.~Segal.
\newblock The communication requirements of efficient allocations and
  supporting prices.
\newblock \emph{J. Econ. Theory}, 129\penalty0 (1):\penalty0 192--224, 2006.
\newblock \doi{10.1016/j.jet.2004.10.007}.

\bibitem[{Paes Leme}(2017)]{PaesLeme17}
R.~{Paes Leme}.
\newblock Gross substitutability: {A}n algorithmic survey.
\newblock \emph{Games Econ. Behav.}, 106:\penalty0 294--316, 2017.
\newblock \doi{10.1016/j.geb.2017.10.016}.

\bibitem[Reiffenh{\"{a}}user(2019)]{Reiffenhauser19}
R.~Reiffenh{\"{a}}user.
\newblock An optimal truthful mechanism for the online weighted bipartite
  matching problem.
\newblock In \emph{Proc.~30th ACM-SIAM SODA}, pages 1982--1993, 2019.
\newblock \doi{10.1137/1.9781611975482.120}.

\bibitem[Rubinstein et~al.(2020)Rubinstein, Wang, and Weinberg]{RubinsteinWW20}
A.~Rubinstein, J.~Z. Wang, and S.~M. Weinberg.
\newblock Optimal single-choice prophet inequalities from samples.
\newblock In \emph{Proc.~11th ITCS}, pages 60:1--60:10, 2020.
\newblock \doi{10.4230/LIPIcs.ITCS.2020.60}.

\bibitem[Rustichini et~al.(1994)Rustichini, Satterwhite, and Williams]{rsw94}
A.~Rustichini, M.~A. Satterwhite, and S.~R. Williams.
\newblock Convergence to efficiency in a simple market with incomplete
  information.
\newblock \emph{Econometrica}, 62\penalty0 (5):\penalty0 1041--1063, 1994.
\newblock URL \url{https://doi.org/10.2307/2951506}.

\bibitem[Satterthwaite and Williams(1989)]{SW89}
M.~A. Satterthwaite and S.~R. Williams.
\newblock The rate of convergence to efficiency in the buyer's bid double
  auction as the market becomes large.
\newblock \emph{Rev. Econ. Stud.}, 56\penalty0 (4):\penalty0 pp. 477--498,
  1989.
\newblock URL \url{https://doi.org/10.2307/2951506}.

\bibitem[Satterwhite and Williams(2002)]{sw02}
M.~A. Satterwhite and S.~R. Williams.
\newblock The optimality of a simple market mechanism.
\newblock \emph{Econometrica}, 70\penalty0 (5):\penalty0 1841--1863, 2002.
\newblock URL \url{https://www.jstor.org/stable/3082022}.

\bibitem[{Segal-Halevi} et~al.(2016){Segal-Halevi}, Hassidim, and
  Aumann]{HHA16}
E.~{Segal-Halevi}, A.~Hassidim, and Y.~Aumann.
\newblock {SBBA}: {A} strongly-budget-balanced double-auction mechanism.
\newblock In \emph{Proc.~9th SAGT}, pages 260--272, 2016.
\newblock \doi{10.1007/978-3-662-53354-3\_21}.

\bibitem[Vickrey(1961)]{vickrey61}
W.~Vickrey.
\newblock Counterspeculation, auctions, and competitive sealed tenders.
\newblock \emph{J. Finance}, 16\penalty0 (1):\penalty0 8--37, 1961.
\newblock \doi{https://doi.org/10.2307/2977633}.

\end{thebibliography}

\clearpage 

\appendix
\section{Omitted Proofs}
\label{app:proofs_of_section3}

\begin{proof}[Proof of \cref{thm:impossibility}]
We concentrate on bilateral trade only, since for richer two-sided problems, we can just assume the social welfare to be dominated by a single buyer-seller pair. Let $\alpha$ be any approximation factor larger than one and assume for contradiction the existence of an $\alpha$-approximative mechanism $M$ as described above. Call $\bP_{v_b,v_s}$ the probability of trade for any reported $v_b, v_s$ following $M$.

We start off by fixing some reported buyer valuation $v_b$, and argue that if the seller reports a value $v_s^{\star}=\frac{v_b}{\alpha^2}$, mechanism $M$ must trade the item from $s$ to $b$ with probability $\bP_{v_b,v_s^{\star}}$ at least $\frac{1}{1+\alpha}$. In fact, since $M$ is an $\alpha$-approximation, we have that the expected social welfare has to be at least a $1/\alpha$ factor of $v_b$:
\begin{align*}
    \bP_{v_b,v_s^{\star}}\cdot v_b +\left(1-\bP_{v_b,v_s^{\star}} \right)\cdot v_s^{\star} \geq  \frac{1}{\alpha} v_b
\end{align*}
which leads to $ \bP_{v_b,v_s^{\star}}  \geq  \frac{1}{ \alpha+1}.$
We will employ this bound on $\bP_{v_b,v_s^{\star}}$ to infer a lower bound for the expected price of trade of the seller $p_s$ in terms of $v_b$ and $\alpha.$ On one hand, when the seller reports $v_s^{\star}$ with her true valuation being $v_s\leq v_s^{\star}$, she gains at least
\begin{align*}
    g^{\star}_s &= \bP_{v_b,v_s^{\star}} \cdot \expect{p_s(v_b,v_s^{\star})|\text{trade}} + (1-\bP_{v_b,v_s^{\star}})\cdot v_s\\
    &\geq \frac{1}{\alpha+1}v_s^{\star}+\left(1-\frac{1}{\alpha+1}\right)v_s \\
    &\geq \frac{v_b}{(\alpha+1)\alpha^2},
\end{align*}
where $p_s(v_b,v_s)\geq v_s$ by IR when a trade occurs.
One the other hand, since $M$ is truthful, reporting $v_s$ gives the seller at least the same gain, i.e.,
$g_s\geq g_s^{\star}$, where
$$
g_s = \bP_{v_b,v_s} \cdot \expect{p_s(v_b,v_s)|\text{trade}} + (1-\bP_{v_b,v_s})\cdot v_s \leq \expect{p_s(v_b,v_s)|\text{trade}}.
$$
Put together, we get for every fixed $v_b$ and $v_s \leq v_b/\alpha^2$,
\begin{align}\label[ineq]{eq:sellerside}
    \expect{p_s(v_b,v_s)|\trade}\geq g_s \geq g_s^{\star} \geq \frac{v_b}{(\alpha+1)\alpha^2}.
\end{align}

Now we change perspective: we fix any $v_s$ and consider any two values $v_b > v_b'$ for 
the buyer. As a first step in this new direction we show that the probability of a trade is a non-decreasing function of the reported buyer valuations. By truthfulness, assuming that the true value is $v_b$ we obtain
$$
(v_b - \expect{p_b(v_b,v_s) \;|\; \trade})\bP_{v_b,v_s} \ge
(v_b - \expect{p_b(v_b',v_s) \;|\; \trade })\bP_{v_b',v_s},
$$
similarly, if the true value is $v_b'$
$$
(v_b' - \expect{p_b(v_b',v_s) \;|\; \trade})\bP_{v_b',v_s} \ge
(v_b' - \expect{p_b(v_b,v_s) \;|\; \trade})\bP_{v_b,v_s}.
$$
Adding both inequalities, we have that $\bP_{v_b,v_s} \ge \bP_{v_b',v_s}$. This implies that, for any fixed $v_s$, there is a fixed value $\bP_{v_s}$ such that $\lim_{v_b \to \infty}\bP_{v_b,v_s} = \bP_{v_s}\leq 1$.

At the same time, we have that  $\expect{p_b(v_b,v_s) \;|\; \trade}$ is also non decreasing in $v_b$, otherwise reporting a higher value would be a beneficial 
deviation for the buyer, increasing the probability of a trade while 
keeping the expected price paid at most as large.

For any $v_s$, and any $\epsilon>0$ arbitrarily close to $0$, let $v_b^{\epsilon,v_s}$ denote a fixed
buyer bid $v_b$ such that the trading probability $\bP_{v_b,v_s}$ for $(v_b,v_s)$ is at least $\bP_{v_s}-\epsilon$.
Then, for a buyer value of $v_b\geq v_b^{\epsilon,v_s}$, truthfulness yields
\begin{align*}
    \bP_{v_s}&\left( v_b-\expect{p_b(v_b,v_s)|\trade} \right)  \geq  \left(\bP_{v_s}-\epsilon \right)\left( v_b-\expect{p_b(v_b^{\epsilon,v_s},v_s)|\trade} \right).
\end{align*}
With some simple algebra, we get that 
\begin{align}
    \notag
     \left(1-\left( \frac{\bP_{v_s}-\epsilon}{\bP_{v_s}} \right)\right)v_b &\geq  \expect{p_b(v_b,v_s)|\trade} -\left( \frac{\bP_{v_s}-\epsilon}{\bP_{v_s}} \right)\expect{p_b(v_b^{\epsilon,v_s},v_s)|\trade} \notag \\ 
    &\geq \expect{p_b(v_b,v_s)|\trade} -\expect{p_b(v_b^{\epsilon,v_s},v_s)|\trade} 
    \label[ineq]{eq:buyerside}
\end{align}
Having shown these inequalities about the expected trading prices for the buyer and seller, we need to combine both. To this end, consider any input $(v_b,v_s)$ to mechanism $M$ with $v_b\gg v_s$, i.e., $v_b \geq \max\{\alpha^2 v_s,v_b^{\epsilon,v_s}\}$. 
Weak budget-balance together with individual rationality implies, using \cref{eq:sellerside} and rearranging the terms in \cref{eq:buyerside}:
\begin{align*}
    \left( 1-\frac{\bP_{v_s}-\epsilon}{\bP_{v_s}} \right) v_b - \expect{p_b(v_b^{\epsilon,v_s},v_s)|\trade } &\geq  \expect{p_b(v_b,v_s)|\trade} \\
    &\geq  \expect{p_s(v_b,v_s)|\trade}\\
    &\geq  \frac{v_b}{(\alpha+1)\alpha^2}.
\end{align*}
Reordering this, we have that the expected price charged from the buyer when she reports $v_b^{\epsilon}$ is at least some factor times $v_b$, for any large $v_b$. More concretely, it holds
\begin{align*}
    \expect{p_b(v_b^{\epsilon,v_s},v_s)|\trade} \geq \left(\frac{1}{(\alpha+1)\alpha^2}-\left( 1-\frac{\bP_{v_s}-\epsilon}{\bP_{v_s}} \right) \right) v_b.
\end{align*}
However, for $\epsilon$ chosen small enough, the right side is clearly approaching infinity when $v_b$ increases, while $v_b^{\epsilon,v_s}$ remains constant since it depends only on $v_s$ but not $v_b$. We have therefore shown that $\expect{p_b( v_b^{\epsilon ,v_s},v_s)|\trade}>v_b^{\epsilon,v_s}$ for appropriate choice of $\epsilon$ and $v_b$, a contradiction to $M$ being IR.
\end{proof}

\end{document}